\documentclass[acmsmall,screen,fleqn]{acmart}
\setcopyright{rightsretained}
\acmPrice{}
\acmDOI{10.1145/3290359}
\acmYear{2019}
\copyrightyear{2019}
\acmJournal{PACMPL}
\acmVolume{3}
\acmNumber{POPL}
\acmArticle{46}
\acmMonth{1}

\usepackage{booktabs}   
\usepackage{subcaption} 
\usepackage[utf8]{inputenc}
\usepackage{amsmath,amssymb,hyperref,array,xcolor,multicol,verbatim}
\usepackage{amsthm}
\usepackage{wasysym}
\usepackage{color}
\usepackage{arydshln}
\usepackage{mathdots}
\usepackage{url}
\usepackage{ latexsym }
\usepackage{algorithm}
\usepackage[inline]{enumitem}

\usepackage{listings}
\usepackage{lstautogobble}
\usepackage{soul}
\usepackage{stmaryrd}
\usepackage{fleqn}
\usepackage{stackengine}
\usepackage{wrapfig}

\usepackage{tikz}
\usetikzlibrary{shapes.misc,arrows.meta}

\theoremstyle{plain}
\newtheorem{theorem}{Theorem}
\theoremstyle{definition}
\newtheorem{definition}{Definition}
\newtheorem{lemma}[theorem]{Lemma}

\newtheorem{proposition}[theorem]{Proposition}
\newtheorem{corollary}[theorem]{Corollary}
\newtheorem{example}{Example}

\newcommand{\mpcomment}[1]{}

\newcommand{\figlabel}[1]{\label{fig:#1}}
\newcommand{\figref}[1]{Figure~\ref{fig:#1}}
\newcommand{\seclabel}[1]{\label{sec:#1}}
\newcommand{\secref}[1]{Section~\ref{sec:#1}}
\newcommand{\exlabel}[1]{\label{ex:#1}}
\newcommand{\exref}[1]{Example~\ref{ex:#1}}
\newcommand{\deflabel}[1]{\label{def:#1}}
\newcommand{\defref}[1]{Definition~\ref{def:#1}}
\newcommand{\thmlabel}[1]{\label{thm:#1}}
\newcommand{\thmref}[1]{Theorem~\ref{thm:#1}}
\newcommand{\proplabel}[1]{\label{prop:#1}}

\newcommand{\lemlabel}[1]{\label{lem:#1}}
\newcommand{\lemref}[1]{Lemma~\ref{lem:#1}}

\newcommand{\nats}{\mathbb{N}}

\newcommand{\Cc}{\mathcal{C}}

\newcommand{\Ff}{\mathcal{F}}
\newcommand{\Rr}{\mathcal{R}}

\newcommand{\Aa}{\mathcal{A}}
\newcommand{\Mm}{\mathcal{M}}

\newcommand{\Ll}{\mathcal{L}}
\newcommand{\Tt}{\mathcal{T}}

\renewcommand{\vec}[1]{{\bf #1}}
\newcommand{\set}[1]{\{#1\}}
\newcommand{\setpred}[2]{\{#1 \,\mid\, #2\}}
\newcommand{\sem}[1]{\llbracket #1 \rrbracket}
\newcommand{\angular}[1]{\langle #1 \rangle}
\def\delequal{\mathrel{\ensurestackMath{\stackon[1pt]{=}{\scriptstyle\Delta}}}}
\renewcommand{\emptyset}{\varnothing}

\newcommand{\xdownarrow}[1]{%
  {\left\downarrow\vbox to #1{}\right.\kern-\nulldelimiterspace}
}
\newcommand{\longproj}[2]{{#1}\xdownarrow{0.27cm}_{#2}}

\newcommand{\stmt}{\angular{stmt}}
\newcommand{\cond}{\angular{cond}}

\newcommand{\code}[1]{\texttt{#1}}
\newcommand{\codekey}[1]{\textbf{#1}}
\newcommand{\cd}[1]{\code{#1}}
\newcommand{\pskip}{\codekey{skip}}
\newcommand{\passume}{\codekey{assume}}
\newcommand{\pif}{\codekey{if}}
\newcommand{\pthen}{\codekey{then}}
\newcommand{\pelse}{\codekey{else}}
\newcommand{\pwhile}{\codekey{while}}
\newcommand{\passign}{:=}
\newcommand{\poutputs}{\Rightarrow}

\newcommand{\exec}{\textsf{Exec}}
\newcommand{\pexec}{\textsf{PExec}}
\newcommand{\comp}{\textsf{Comp}}
\newcommand{\Terms}{\textsf{Terms}}

\newcommand{\init}[1]{\widehat{#1}}

\newcommand{\dblqt}[1]{\text{``}#1\text{''}}

\newcommand{\fixperm}[1]{\angular{#1}}
\newcommand{\out}[1]{\textbf{\textsf{o}}_{#1}}
\newcommand{\pgm}{\angular{pgm}}
\newcommand{\pcall}{\codekey{call}}
\newcommand{\preturn}{\codekey{return}}
\newcommand{\goesto}{\rightarrow}

\newcommand{\inc}{\mathsf{INC}}
\newcommand{\dec}{\mathsf{DEC}}
\newcommand{\chkzero}{\mathsf{CHECK}}
\newcommand{\halt}{\mathsf{HALT}}
\newcommand{\curr}{\text{curr}}

\newcommand{\congcl}[1]{\cong_{#1}}

\newcommand{\eqcl}[2]{[#1]_{#2}}
\newcommand{\termsem}[2]{\sem{#1}_{#2}}
\newcommand{\termmod}[1]{\Tt(#1)}
\newcommand{\fals}{\bot}
\newcommand{\tru}{\top}
\newcommand{\undf}{\mathsf{undef}}
\newcommand{\proj}[2]{#1\!\!\downharpoonright_{#2}}
\newcommand{\reject}{q_\mathsf{reject}}

\newcommand{\feas}[1]{#1_{\mathsf{fs}}}
\newcommand{\coh}[1]{#1_{\mathsf{cc}}}
\newcommand{\rfeas}[1]{#1_{\mathsf{rfs}}}

\newcommand{\pspc}{\mathsf{PSPACE}}

\newcommand{\drawline}{\raisebox{2pt}{\scalebox{0.7}{\tikz{\draw[-, thick](0,0) -- (7mm,0);}}}}
\newcommand{\drawdash}{\raisebox{2pt}{\scalebox{0.7}{\tikz{\draw[-, thick, dashed](0,0) -- (7mm,0);}}}}
\newcommand{\drawdirectedline}{\raisebox{0pt}{\scalebox{0.7}{\tikz{\draw[-{Latex[length=2mm, width=2mm]}, thick](0,0) -- (7mm,0);}}}}

\begin{document}

\title{Decidable Verification of Uninterpreted Programs}


\author{Umang Mathur}
\orcid{0000-0002-7610-0660}             
\affiliation{
  \department{Department of Computer Science}              
  \institution{University of Illinois, Urbana Champaign}            
  \country{USA}                    
}
\email{umathur3@illinois.edu}         

\author{P. Madhusudan}
\affiliation{
  \department{Department of Computer Science}              
  \institution{University of Illinois, Urbana Champaign}            
  \country{USA}                    
}
\email{madhu@illinois.edu}          

\author{Mahesh Viswanathan}
\affiliation{
  \department{Department of Computer Science}              
  \institution{University of Illinois, Urbana Champaign}            
  \country{USA}                    
}
\email{vmahesh@illinois.edu}         

\begin{abstract}
We study the problem of completely automatically verifying uninterpreted programs---programs that work over arbitrary data models that provide an interpretation for the constants, functions and relations the program uses. The verification problem asks whether a given program satisfies a postcondition written using quantifier-free formulas with equality on the final state, with no loop invariants, contracts, etc. being provided. We show that this problem is undecidable in general. The main contribution of this paper is a subclass of programs, called \emph{coherent programs} that admits decidable verification, and can be decided in {\sc Pspace}. We then extend this class of programs to classes of programs that are $k$-coherent, where $k \in \mathbb{N}$, obtained by (automatically) adding $k$ ghost variables and assignments that make them coherent. We also extend the decidability result to programs with recursive function calls and prove several undecidability results that show why our restrictions to obtain decidability seem necessary.

\end{abstract}


\maketitle


\section{Introduction}
\seclabel{intro}

Completely automatic verification of programs is almost always
undecidable.  The class of sequential programs, with and without
recursive functions, admits a decidable verification problem when the
state-space of variables/configurations is \emph{finite}, and this has
been the cornerstone on which several fully automated verification
techniques have been based, including predicate abstraction, and
model-checking.  However, when variables range over infinite domains,
verification almost inevitably is undecidable.  For example, even for
programs manipulating natural numbers with increment and decrement
operators, and checks for equality, program verification is
undecidable.

In this paper, we investigate classes of programs over 
\emph{uninterpreted} functions and relations over infinite 
domains that admit, surprisingly, a decidable verification problem 
(with no user help whatsoever, not even in terms of inductive loop 
invariants or pre/post conditions). 

A program can be viewed as working over a data-domain that consists of
constants, functions and relations. For example, a program
manipulating integers works on a data-model that provides constants
like $0, 1$, functions like $+, -$, and relations like $\leq$, where
there is an implicit assumption on the meaning of these constants,
functions, and relations. Programs over uninterpreted data models work
over \emph{arbitrary} data-models, where the interpretation of
functions and relations are not restricted in any way, except of
course that equality is a congruence with respect to their
interpretations (e.g., if $x=y$, then $f(x)=f(y)$, no matter what the
interpretation of $f$ is).  A program satisfies its assertions over
uninterpreted data models if it satisfies the assertions when working
over \emph{all} data-models.

The theory of uninterpreted functions is a theory that only has the
congruence axioms, and is an important theory both from a theoretical
and practical standpoint.  Classical logic such as G\"odel's (weak)
completeness theorem are formulated for such theories.  And in
verification, when inductive loop invariants are given, verification
conditions are often formulated as formulas in SMT theories, where the
theory of uninterpreted functions is an important theory used to model
memory, pointers in heaps, arrays, and mathematical specifications.
In particular, the quantifier-free logic of uninterpreted functions is
decidable and amenable to Nelson-Oppen combination with other
theories, making it a prime theory in SMT solvers.

We show, perhaps unsurprisingly, that verification of uninterpreted
programs is undecidable.  The main contribution of this paper is to
identify a class of programs, called \emph{coherent} programs, for
which verification is decidable.

Program executions can be viewed abstractly as computing \emph{terms}
conditioned on \emph{assumptions over terms}.  Assignments apply
functions of the underlying data-domain and hence the value of a
variable at any point in an execution can be seen as computing a term
in the underlying data-model.  Conditional checks executed by the
program can be seen as assumptions the program makes regarding the
relations that hold between terms in the data-model.  For example,
after an execution of the statements $x\passign y; x\passign x+1;
\passume (x>0); z\passign x*y$, the program variable $z$ corresponds
to the term $(\init{y}+1)*\init{y}$, and the execution makes the
assumption that $\init{y}+1>0$, where $\init{y}$ is the value of
variable $y$ at the start of the execution.  A coherent program has
only executions where the following two properties hold.  The first is
the \emph{memoizing property} that says that when a term is
\emph{recomputed} by the execution, then some variable of the program
already has the same term (or perhaps, a different term that is
equivalent to it, modulo the assumptions seen so far in the
execution).  The second property, called \emph{early assumes} says,
intuitively, that when an assumption of equality between variables $x$
and $y$ is made by a program, superterms of the terms stored in
variables $x$ and $y$ computed by the program must be stored in one of
the current variables.

We show that the notion of coherence effectively skirts undecidability
of program verification.  Both notions of memoizing and early-assumes
require variables to store certain computed terms in the current set
of variables.  This notion in fact is closely related to \emph{bounded
  path-width} of the computational graph of terms computed by the
program; bounded path-width and bounded tree-width are graph-theoretic
properties exploited by many decision procedures of graphs for
decidability of MSO and for efficient
algorithms~\cite{courcelle,seese}, and there have been several recent
works where they have been useful in finding decidable problems in
verification~\cite{madhu2011,chatterjee2016,chatterjee2015}.

Our decidability procedure is automata-theoretic.  We show that
coherent programs generate \emph{regular} sets of coherent executions,
and we show how to construct automata that check whether an execution
satisfies a post-condition assertion written in quantifier-free theory
of equality.  The automaton works by computing the \emph{congruence
  closure} of terms defined by equality assumptions in the execution,
checking that the disequality assumptions are met, while maintaining
this information only on the bounded window of terms corresponding to
the current valuation of variables of the program. In fact, the
automaton can be viewed as a \emph{streaming congruence closure algorithm} that computes
the congruence closure on the moving window of terms computed by the
program.  The assumption of coherence is what allows us to build such
a streaming algorithm.  We show that if either the \emph{memoizing}
assumption or the \emph{early-assumes} assumption is relaxed, the
verification problem becomes undecidable, arguing for the necessity of
these assumptions. 

The second contribution of this paper is a decidability result that
extends the first result to a larger class of programs --- those
uninterpreted programs that can be \emph{converted} to coherent ones.
A program may not be coherent because of an execution that either
recomputes a term when no variable holds that term, or makes an
assumption on a term whose superterm has been previously computed but
later over-written.  However, if the program was given access to more
variables, it could keep the required term in an auxiliary variable to
meet the coherence requirement.  We define a natural notion of
$k$-coherent executions --- executions that can be made coherent by
adding $k$ \emph{ghost variables} that are write-only and assigned at
appropriate times.  We show that programs that generate $k$-coherent
executions also admit a decidable verification problem.  The notion of
$k$-coherence is again related to path-width --- instead of demanding
that executions have path-width $|V|$, we allow them to have
path-width $|V|+k$, where $V$ is the set of program variables.

We also show that $k$-coherence is a decidable property.  Given a
program and $k \in \nats$, we can decide if its executions can be made
$k$-coherent. Notice that when $k = 0$, $k$-coherence is simply
coherence, and so these results imply the decidability of coherence as
well.  And if they can, we can automatically build a regular
collection of coherent executions that automatically add the ghost
variable assignments.  This result enables us to verify programs by
simply providing a program and a budget $k$ (perhaps iteratively
increased), and automatically check whether the program's executions
can be made $k$-coherent, and if so, perform automatic verification
for it.

The third contribution of this paper is an extension of the 
above results to programs with recursive function calls. 
We show that we can build \emph{visibly pushdown automata} (VPA) 
that read coherent executions of programs 
with function calls, and compute congruence closure effectively. 
Intersecting such a VPA with the VPA accepting 
the program executions and checking emptiness of the resulting 
VPA gives the decidability result.
We also provide the extension of verification to $k$-coherent recursive programs. 

To the best of our knowledge, the results here present the first
interesting class of sequential programs over infinite data-domains
for which verification is decidable\footnote{There are some
  automata-theoretic results that can be interpreted as decidability
  results for sequential programs; but these are typically programs
  reading streaming data or programs that allow very restricted ways
  of manipulating counters, which are not natural classes for software
  verification. See~\secref{related} for a detailed discussion.}.  

The main contributions of this paper
are:
\begin{itemize}
 \item We show verification of uninterpreted programs (without function calls) is undecidable.
 \item We introduce a notion of coherent programs and show verification of coherent uninterpreted programs (without
 function calls) is decidable and
 is {\sc Pspace}-complete.
 \item We introduce a notion of $k$-coherent programs, for any $k$. 
 We show that given a program (without function calls) and a constant $k$, 
 we can decide if it is $k$-coherent; and if it is, decide verification for it.
\item We prove the above results for programs with (recursive) function calls, showing decidability and {\sc Exptime}-completeness. 
\end{itemize}

The paper is structured as follows. 
\secref{definitions} introduces uninterpreted programs and their 
verification problem, and summarizes the main results of the paper. 
\secref{coherent-ver} contains our main technical result and is devoted 
to coherent programs and the decision procedure for verifying them, 
as well as the decision procedure for recognizing coherent programs. 
In \secref{undecidability}, we show our undecidability results for general 
progams as well as programs that satisfy only one of the conditions of the coherence definition. 
\secref{kcoherence} consists of our decidability results for $k$-coherent programs. 
In \secref{vpa} we extend our results to recursive programs. 
Related work discussion can be found in \secref{related} and concluding 
remarks in \secref{conclusions} where we also discuss possible extensions and 
applications of our results.
We refer the reader to the appendix
for detailed proofs of the results presented.


\section{The Verification Problem and Summary of Results}
\seclabel{definitions}

In this paper we investigate the verification problem for imperative
programs, where expressions in assignments and conditions involve
uninterpreted functions and relations. We, therefore, begin by
defining the syntax and semantics of the class of programs we study,
and then conclude the section by giving an overview of our main
results; the details of these results will be presented in subsequent
sections.

Let us begin by recalling some classical definitions about first order
structures. A (finite) first order \emph{signature} $\Sigma$ is a tuple $(\Cc,
\Ff, \Rr)$, where $\Cc$, $\Ff$, and $\Rr$ are finite sets of
constants, function symbols, and relation symbols, respectively. Each
function symbol and relation symbol is implicitly associated with an
arity in $\nats_{>0}$. A first order signature is 
\emph{algebraic} if there are no relation symbols, i.e., $\Rr =
\emptyset$. We will denote an algebraic signature as $\Sigma = (\Cc,
\Ff)$ instead of $\Sigma = (\Cc, \Ff, \emptyset)$. An \emph{algebra}
or \emph{data model} for an algebraic signature $\Sigma = (\Cc, \Ff)$,
is $\Mm = (U, \setpred{\sem{c}}{c \in \Cc}, \setpred{\sem{f}}{f \in
  \Ff})$ which consists of a universe $U$ and an interpretation for
each constant and function symbol in the signature. The set of
(ground) \emph{terms} are those that can be built using the constants
in $\Cc$ and the function symbols in $\Ff$; inductively, it is the set
containing $\Cc$, and if $f$ is an $m$-ary function symbol, and $t_1,
\ldots t_m$ are terms, then $f(t_1,\ldots t_m)$ is a term. We will
denote the set of terms as $\Terms_\Sigma$ or simply $\Terms$, 
since the signature $\Sigma$ will often be clear from the context. Given a term $t$ and a data model
$\Mm$, the \emph{interpretation} of $t$ (or the \emph{value} that $t$ evaluates to) in $\Mm$ will be denoted
by $\termsem{t}{\Mm}$.

\subsection{Programs}
\seclabel{pgmsyntax}

Our imperative programs will use a finite set of variables to store
information during a computation. Let us fix $V = \{v_1,\ldots v_r\}$
to be this finite set of variables. These programs will use function
symbols and relation symbols from a first order signature $\Sigma =
(\Cc, \Ff, \Rr)$ to manipulate values stored in the variables. We will
assume, without loss of generality, that the first order signature has
constant symbols that correspond to the \emph{initial values} of each
variable at the begining of the computation. More precisely, let
$\init{V} = \setpred{\init{x}}{x \in V} \subseteq \Cc$ represent the
initial values for each variable of the program. The syntax of
programs is given by the following grammar.
\begin{align*}
\stmt ::=& 
\,\,  \pskip \, 
\mid \, x \passign c \,
\mid \, x \passign y \, 
\mid \, x \passign f(\vec{z}) \, 
\mid \, \passume \, (\cond) \,
\mid \, \stmt \, ;\, \stmt \\
&
\mid \, \pif \, (\cond) \, \pthen \, \stmt \, \pelse \, \stmt \,
\mid \, \pwhile \, (\cond) \, \stmt
\end{align*}
\begin{align*}
\cond ::=& 
\, x = y \,
\mid \, x = c \,
\mid \, c = d \,
\mid \, R(\vec{z}) \,
\mid \, \cond \lor \cond \,
\mid \, \neg \cond
\end{align*}
Here, $f \in \Ff$, $R \in \Rr$, $c, d \in \Cc$, $x,y \in V$, and
$\vec{z}$ is a tuple of variables in $V$ and constants in $\Cc$.

The constructs above define a simple class of programs with
conditionals and loops. Here, `$\passign$' denotes the assignment
operator, `$;$' stands for sequencing of programs, $\pskip$ is a ``do
nothing'' statement, $\pif-\pthen-\pelse$ is a construct for conditional
statements and $\pwhile$ is our looping construct.  We will also use
the shorthand `$\pif \, (\cond) \, \pthen \, \stmt$' as syntactic sugar for
`$\pif \, (\cond) \, \pthen \, \stmt \, \pelse \, \pskip$'.  The
conditionals can be equality ($=$) atoms, predicates defined by
relations $(R(\cdot))$, and boolean combinations $(\lor, \neg)$ of
other conditionals.  Formally, the semantics of the program depends on
an underlying data model that provides a universe, and meaning for
functions, relations, and constants; we will define this precisely
in~\secref{pgmsemantics}.

The conditionals in the above syntax involve Boolean combinations of
equalities as well as relations over variables and constants.
However, for technical simplicity and without loss of generality, we
disallow relations entirely.  Note that a relation $R$ of arity $m$
can be modeled by fixing a new constant $\top$ and introducing a new
function $f_R$ of arity $m$ and a variable $b_R$.  Then, each time we
see $R(\vec{z})$, we add the assignment statement $b_R \passign
f_R(\vec{z})$ and replace the occurrence of $R(\vec{z})$ by the
conditional `$b_R = \top$'.  Also, Boolean combinations of conditions
can be modeled using the $\pif-\pthen-\pelse$ construct.
%
%
Constant symbols used in conditionals and assignments can also be removed
simply by using a variable in the program that is not modified in any
way by the program. Hence we will avoid the use of constant symbols as
well in the program syntax. Henceforth, without loss generality, we
can assume that our first order signature $\Sigma$ is algebraic ($\Rr
= \emptyset$), constant symbols do not appear in any of the program
expressions, and our programs have conditionals only of the form $x=y$
or $x \not = y$.


\begin{figure}[t]
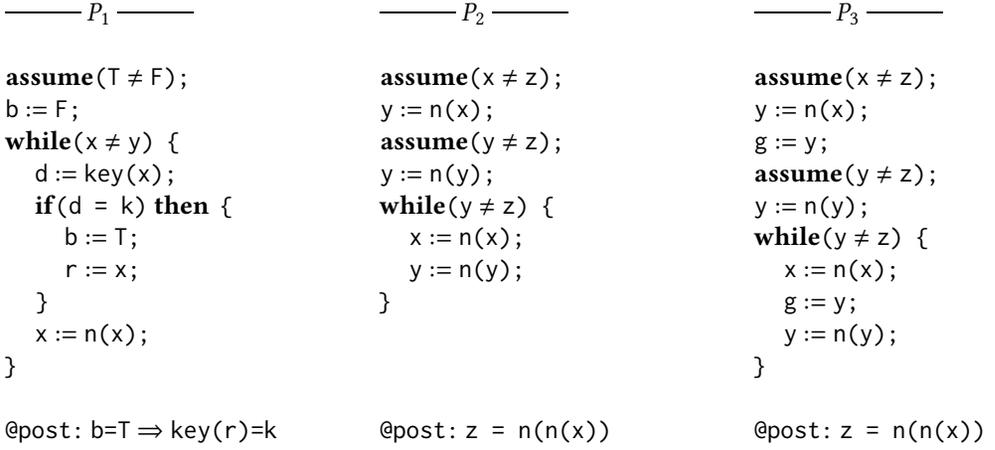

\noindent\makebox[1.2\textwidth][c]{
\begin{minipage}{1.1\textwidth}
\begin{minipage}[H]{0.32\textwidth}
\rule[1mm]{1cm}{0.4pt} $P_1$ \rule[1mm]{1cm}{0.4pt} \\\\
\passume \code{(T} $\neq$ \code{F);} \\
\code{b} $\passign$ \code{F;}\\
\pwhile \code{(x} $\neq$ \code{y) \{} \\ 
\rule[1mm]{0.3cm}{0pt}
\code{d} $\passign$ \code{key(x);} \\
\rule[1mm]{0.3cm}{0pt}
\pif \code{(d = k)} \pthen\, \code{\{} \\
\rule[1mm]{0.7cm}{0pt}
\code{b} $\passign$ \code{T;} \\ 
\rule[1mm]{0.7cm}{0pt}
\code{r} $\passign$ \code{x;} \\ 
\rule[1mm]{0.3cm}{0pt}
\code{\}}\\
\rule[1mm]{0.3cm}{0pt}
\code{x} $\passign$ \code{n(x);} \\ 
\code{\}} \\\\
\code{@post:} \code{b=T} $\Rightarrow$ \code{key(r)=k}
\end{minipage}
\begin{minipage}[h]{0.32\textwidth}
\rule[1mm]{1cm}{0.4pt} $P_2$ \rule[1mm]{1cm}{0.4pt} \\\\
\passume \code{(x} $\neq$ \code{z);} \\
\code{y} $\passign$ \code{n(x);}\\
\passume \code{(y} $\neq$ \code{z);} \\
\code{y} $\passign$ \code{n(y);}\\
\pwhile \code{(y} $\neq$ \code{z) \{} \\ 
\rule[1mm]{0.3cm}{0pt}
\code{x} $\passign$ \code{n(x);} \\ 
\rule[1mm]{0.3cm}{0pt}
\code{y} $\passign$ \code{n(y);} \\ 
\code{\}} \\\\\\\\
\code{@post:} \code{z = n(n(x))}
\end{minipage}
\begin{minipage}[h]{0.32\textwidth}
\rule[1mm]{1cm}{0.4pt} $P_3$ \rule[1mm]{1cm}{0.4pt} \\\\
\passume \code{(x} $\neq$ \code{z);} \\
\code{y} $\passign$ \code{n(x);}\\
\code{g} $\passign$ \code{y;}\\
\passume \code{(y} $\neq$ \code{z);} \\
\code{y} $\passign$ \code{n(y);}\\
\pwhile \code{(y} $\neq$ \code{z) \{} \\ 
\rule[1mm]{0.3cm}{0pt}
\code{x} $\passign$ \code{n(x);} \\ 
\rule[1mm]{0.3cm}{0pt}
\code{g} $\passign$ \code{y;}\\
\rule[1mm]{0.3cm}{0pt}
\code{y} $\passign$ \code{n(y);} \\ 
\code{\}} \\\\
\code{@post:} \code{z = n(n(x))}
\end{minipage}
\end{minipage}
}
\caption{Examples of Uninterpreted Programs; $P_1$ and $P_3$ are coherent, $P_2$ is not coherent}
\figlabel{example1}
\end{figure}

\begin{example}
Consider the uninterpreted program $P_1$ in~\figref{example1}. 
The program works on any first-order model that has an 
interpretation for the unary functions $\code{n}$ and $\code{key}$, 
and an initial interpretation of the variables 
\code{T}, \code{F}, \code{x}, \code{y} and \code{k}. 
The program is similar to a program that searches 
whether a list segment from $\code{x}$ to $\code{y}$
contains a key $\code{k}$. 
However, in the program above, the
functions $\code{n}$ and $\code{key}$ are uninterpreted, 
and we allow all possible models on which the program can work.
Note that if and when the program terminates, 
we know that if $\code{b} = \code{T}$, then there is an element reachable 
from $\code{x}$ before reaching $\code{y}$ such that 
$\code{key}$ applied to that node is equal to $\code{k}$. 
Note that we are modeling $\code{T}$ and $\code{F}$, 
which are Boolean constants, as variables in the program 
(assuming that they are different elements in the model). 

Programs $P_2$ and $P_3$ in~\figref{example1} are also uninterpreted
programs, and resemble programs that given a linked list segment from
$\code{x}$ to $\code{z}$, finds the node that is two nodes before the
node $\code{z}$ (i.e., find the node $u$ such that
$\code{n}(\code{n}(u))=\code{z}$).
\end{example}

\subsection{Executions}
\seclabel{executions}

\begin{definition}[Executions]
An execution over a finite set of variables $V$ is a word over the alphabet
$\Pi = \setpred{ \dblqt{x \passign y}, \dblqt{x \passign f(\vec{z})},  \dblqt{\passume (x=y)}, \dblqt{\passume (x\neq y)}}{x, y, \vec{z} \textit{~are~in~} V}$.
\end{definition}

We use quotes around letters for readability, and may sometimes skip them.

\begin{definition}[Complete and Partial Executions of a program]
\emph{Complete executions} of programs that manipulate a set of variables $V$ are executions over $V$ defined formally as follows:
\begin{align*}
\begin{array}{rcl}
\exec(\pskip) &=& \epsilon  \\
\exec(x \passign y) &=& \dblqt{x \passign y} \\
\exec(x \passign f(\vec{z})) &=& \dblqt{x \passign f(\vec{z})} \\
\exec(\passume(c)) &=& \dblqt{\passume(c)} \\
\exec(\pif\,c \, \pthen\, s_1\,  \pelse\, s_2~~) &=& \dblqt{\passume(c)} \cdot \exec(s_1)
 \cup \dblqt{\passume(\neg c)} \cdot \exec(s_2) \\
\exec(~~s_1 ; s_2~~) &=& \exec(s_1) \cdot \exec(s_2) \\
\exec(\pwhile\, c\, \{s\}~~) &=&
        [\dblqt{\passume(c)} \cdot \exec(s_1)]^* \cdot \dblqt{\passume(\neg c)}
\end{array}
\end{align*}
Here, $c$ is a conditional of the form $x = y$ or $x \neq y$, where $x, y\in V$.

The set of \emph{partial executions}, denoted by $\pexec(s)$, is the set of prefixes of complete executions in $\exec(s)$.

\end{definition}

\begin{example}
\exlabel{exec}
 For the example program $P_1$ in~\figref{example1}, 
the following word $\rho$ 
\begin{align*}
\rho \delequal & 
\, \passume (\code{T} \neq \code{F}) 
\cdot \code{b} \passign \code{F} 
\cdot \passume (\code{x} \neq \code{y})
\cdot \code{d}\passign \code{key(x)} 
\cdot \passume (\code{d} \neq  \code{k}) 
\cdot \cd{x} \passign \cd{n(x)} \\
&
\cdot \passume (\cd{x} \neq  \cd{y}) 
\cdot \cd{d} \passign \cd{key(x)} 
\cdot \passume (\cd{d} = \cd{k}) 
\cdot \cd{b} \passign \cd{T} 
\cdot \cd{r} \passign \cd{x} 
\cdot \cd{x} \passign \cd{n(x)} 
\end{align*}
is a partial execution of $P_1$ and the word $\rho_1 =\rho \cdot
\passume (\cd{x} = \cd{y})$ is a complete execution.
\end{example}

Our notion of executions is more syntactic than semantic. 
In other words, we do not insist that executions are 
\emph{feasible} over any data model.
For example, the word $\passume (x=y) \cdot \passume (x \neq y) \cdot x \passign f(x)$ is an execution though it is not feasible over any data model.  
Note also that the complete executions of a program capture
(syntactically) \emph{terminating} computations, i.e., 
computations that run through the entire program.

It is easy to see that an NFA accepting $\exec(s)$
(as well as for $\pexec(s)$)
of size linear in $s$, for any program $s$, 
can be easily constructed in polynomial time from $s$, 
using the definitions above 
and a standard translation of regular expressions to NFAs.

\begin{example} 
For the program $P_1$ in~\figref{example1}, its set of executions is
given by the following regular expression
\begin{align*}
\passume(\cd{T} \neq \cd{F}) \cdot \cd{b}\passign\cd{F} \cdot R \cdot \passume(\cd{x}=\cd{y})
\end{align*}
where $R$ is the regular expression
\begin{align*}
\left[ \passume(\cd{x} \neq \cd{y}) \cdot \cd{d} \passign \cd{key(x)} 
\cdot \left( \passume(\cd{d} \neq  \cd{k}) 
+  \passume(\cd{d} = \cd{k}) \cdot \cd{b}\passign \cd{T} \cdot \cd{r}\passign \cd{x} \right) 
\cdot \cd{x} \passign \cd{n(x)} \right]^*
\end{align*}

\end{example}

\subsection{Semantics of Programs and The Verification Problem}
\seclabel{pgmsemantics}

\subsection*{Terms Computed by an Execution}
We now define the set of terms computed by executions of a program
over variables $V$.  The idea is to capture the term computed for each
variable at the end of an execution. Recall that $\init{V} =
\setpred{\init{x}}{x \in V}$ is the set of constant symbols that
denote the initial values of the variables in $V$ when the execution
starts, i.e., $\init{x}$ denotes the initial value of variable
$x$, etc. Recall that $\Terms$ are the set of all terms over signature
$\Sigma$.  Let $\Pi = \setpred{ \dblqt{x \passign y}, \dblqt{x
    \passign f(\vec{z})}, \dblqt{\passume (x=y)}, \dblqt{\passume
    (x\neq y)}}{x, y, \vec{z} \textit{~are~in~} V}$ be the alphabet of
executions.

\begin{definition}
\deflabel{computation}
The term assigned to a variable $x$ after some partial execution is
captured using the function $\comp : \Pi^* \times V \to \Terms$
defined inductively as follows.
\begin{flalign*}
\begin{array}{rcll}
 \comp(\epsilon, x) \!\! &=& \!\! \init{x} & \text{ for each } x \in V \\
 \comp(\rho \cdot \dblqt{x \passign y}, x) \!\! &=& \!\! \comp(\rho, y) \\
 \comp(\rho \cdot \dblqt{x \passign y}, x') \!\! &=& \!\! \comp(\rho, x') &x' \neq x \\
 \comp(\rho \cdot \dblqt{x \passign f(\vec{z})}, x) \!\! &=& \!\! f(\comp(\rho, z_1), \ldots, \comp(\rho, z_r)) & \text{where } \vec{z} = (z_1, \ldots, z_r) \\
 \comp(\rho \cdot \dblqt{x \passign f(\vec{z})}, x') \!\! &=& \!\! \comp(\rho, x') & x' \neq x \\
 \comp(\rho \cdot \dblqt{\passume(y = z)}, x) \!\! &=& \!\! \comp(\rho, x) & \text{ for each } x \in V \\
 \comp(\rho \cdot \dblqt{\passume(y \neq z)}, x) \!\! &=& \!\! \comp(\rho, x) & \text{ for each } x \in V 
 \end{array}
\end{flalign*}

The set of \emph{terms computed} by an execution $\rho$ is $ \Terms(\rho) =
\bigcup\limits_{\substack{\rho' \text{ is a prefix of } \rho,\\ x \in V}}
\comp(\rho',x).$

\end{definition}

Notice that the terms computed by an execution are independent of
the $\passume$ statements in the execution
and depend only on the assignment statements.

\begin{example}
Consider the execution (from~\exref{exec}) below
\begin{align*}
 \rho_1 \delequal
 & 
\, \passume (\code{T} \neq \code{F}) 
\cdot \code{b} \passign \code{F} 
\cdot \passume (\code{x} \neq \code{y})
\cdot \code{d}\passign \code{key(x)} 
\cdot \passume (\code{d} \neq  \code{k}) 
\cdot \cd{x} \passign \cd{n(x)} \\ 
& \cdot \passume (\cd{x} \neq  \cd{y}) 
\cdot \cd{d} \passign \cd{key(x)} 
\cdot \passume (\cd{d} = \cd{k}) 
\cdot \cd{b} \passign \cd{T} 
\cdot \cd{r} \passign \cd{x}
\cdot \cd{x} \passign \cd{n(x)} 
\cdot \passume (x = y)
\end{align*}



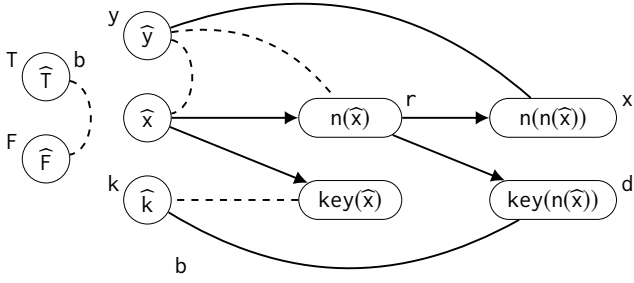
\begin{figure}[t]
\scalebox{0.9}{
\begin{tikzpicture}
\node (y) at (0.5,0.25) [rounded rectangle] {$\cd{y}$};
\node (hat-y) at (1,0) [draw, circle, minimum width = 2em] {$\init{\cd{y}}$};
\node (hat-x) at (1,-1.2) [draw, circle, minimum width = 2em] {$\init{\cd{x}}$};
\node (y) at (0.5,-2.15) [rounded rectangle] {$\cd{k}$};
\node (hat-k) at (1,-2.4) [draw, circle, minimum width = 2em] {$\init{\cd{k}}$};
\node (b) at (1.5,-3.35) [rounded rectangle] {$\cd{b}$};
\node (b) at (0,-0.35) [rounded rectangle] {$\cd{b}$};
\node (T) at (-1,-0.35) [rounded rectangle] {$\cd{T}$};
\node (hat-T) at (-0.5,-0.6) [draw, circle, minimum width = 2em] {$\init{\cd{T}}$};
\node (F) at (-1,-1.55) [rounded rectangle] {$\cd{F}$};
\node (hat-F) at (-0.5,-1.8) [draw, circle, minimum width = 2em] {$\init{\cd{F}}$};

\node (r) at (4.9,-0.95) [rounded rectangle] {$\cd{r}$};
\node (n-hat-x) at (4,-1.2) [draw, rounded rectangle, minimum width = 5em] {$\cd{n}(\init{\cd{x}})$};
\node (x) at (8.1,-0.95) [rounded rectangle] {$\cd{x}$};
\node (n-n-hat-x) at (7,-1.2) [draw, rounded rectangle, minimum width = 6em] {$\cd{n}(\cd{n}(\init{\cd{x}}))$};
\node (key-hat-x) at (4,-2.4) [draw, rounded rectangle, minimum width = 5em] {$\cd{key}(\init{\cd{x}})$};
\node (d) at (8.1,-2.15) [rounded rectangle] {$\cd{d}$};
\node (key-n-hat-x) at (7,-2.4) [draw, rounded rectangle, minimum width = 6em] {$\cd{key}(\cd{n}(\init{\cd{x}}))$};

\draw (hat-y) edge[-, bend left, thick] (n-n-hat-x);
\draw (hat-k) edge[-, bend right, thick] (key-n-hat-x);

\draw (hat-y) edge[-, dashed, bend left, thick] (n-hat-x);
\draw (hat-y) edge[-, dashed, bend left = 80, thick] (hat-x);
\draw (hat-T) edge[-, dashed, bend left = 80, thick] (hat-F);
\draw (key-hat-x) edge[-, dashed, thick] (hat-k);

\draw (hat-x) edge[-{Latex[length=2mm, width=2mm]}, thick] (n-hat-x);
\draw (hat-x) edge[-{Latex[length=2mm, width=2mm]}, thick] (key-hat-x);
\draw (n-hat-x) edge[-{Latex[length=2mm, width=2mm]}, thick] (n-n-hat-x);
\draw (n-hat-x) edge[-{Latex[length=2mm, width=2mm]}, thick] (key-n-hat-x);
\end{tikzpicture}
}
\caption{Computation Graph of $\rho_1$. 
Nodes represent terms computed in $\rho_1$.
Directed edges (\protect\drawdirectedline) represent immediate subterm relation.
Nodes are labelled by variables that correspond to the terms denoted by nodes.
Undirected solid lines (\protect\drawline)~denote equalities and dashed lines (\protect\drawdash)~denote disequalities seen in $\rho_1$}
\figlabel{comp_graph}
\end{figure}

For this execution, the set of terms computed
can be visualized by the computation graph in~\figref{comp_graph}.
Here, the nodes represent the various terms computed by the program, 
the solid directed edges represent the immediate subterm relation, 
the solid lines represent the assumptions of equality made 
in the execution on terms, 
and the dashed lines represent the assumptions of dis-equality 
made by the execution. 
The labels on nodes represent the variables that evaluate 
to the terms at the end of the execution.

Hence, we have
$\comp(\rho_1, \cd{x}) = \cd{n}(\cd{n}(\init{\cd{x}}))$, 
$\comp(\rho_1, \cd{d}) = \cd{key}(\cd{n}(\init{\cd{x}}))$,
$\comp(\rho_1, \cd{b}) = \init{\cd{T}}$,
$\comp(\rho_1, \cd{y}) = \init{\cd{y}}$,
$\comp(\rho_1, \cd{k}) = \init{\cd{k}}$,
$\comp(\rho_1, \cd{T}) = \init{\cd{T}}$,
$\comp(\rho_1, \cd{F}) = \init{\cd{F}}$, and
$\comp(\rho_1, \cd{r}) = \cd{n}(\init{\cd{x}})$.
\end{example}

\subsection*{Equality and Disequality Assumptions of an Execution}

Though the $\passume$ statements in an execution do not influence the
terms that are assigned to any variable, they play a role in defining
the semantics of the program. The equalities and disequalities
appearing in $\passume$ statements must hold in a given data model, for
the execution to be feasible. We, therefore, identify what these are.

For an (partial) execution $\rho$, let us first define the set of
\emph{equality assumes} that $\rho$ makes on terms.  Formally, for any
execution $\rho$, the set of equality assumes defined by $\rho$,
called $\alpha(\rho)$, is a subset of $\Terms(\rho) \times
\Terms(\rho)$ defined as follows.
\begin{align*}
\begin{array}{rcl}
  \alpha(\epsilon) 
  & = & \emptyset\\
  \alpha(\rho \cdot a) 
  & = &
  \begin{cases}
  \alpha(\rho) \cup \set{(\comp(\rho,x), \comp(\rho,y)) } & \text{ if } a \text{ is } \dblqt{\passume (x = y)} \\
  \alpha(\rho) & \text{otherwise}
  \end{cases}
 \end{array}
\end{align*}

The set of disequality assumes, $\beta(\rho)$, can be similarly defined
inductively.
\begin{align*}
\begin{array}{rcl}
  \beta(\epsilon) 
  & = & \emptyset\\
  \beta(\rho \cdot a) 
  & = &
  \begin{cases}
  \beta(\rho) \cup \set{(\comp(\rho,x), \comp(\rho,y)) }  & \text{ if } a \text{ is } \dblqt{\passume (x \neq y)} \\
  \beta(\rho)  & \text{otherwise}
  \end{cases}
 \end{array}
\end{align*}

\begin{example}
Consider the execution (from~\exref{exec}) below
\begin{align*}
 \rho_1 \delequal
 & 
\, \passume (\code{T} \neq \code{F}) 
\cdot \code{b} \passign \code{F} 
\cdot \passume (\code{x} \neq \code{y})
\cdot \code{d}\passign \code{key(x)} 
\cdot \passume (\code{d} \neq  \code{k}) 
\cdot \cd{x} \passign \cd{n(x)} \\ 
& \cdot \passume (\cd{x} \neq  \cd{y}) 
\cdot \cd{d} \passign \cd{key(x)} 
\cdot \passume (\cd{d} = \cd{k}) 
\cdot \cd{b} \passign \cd{T} 
\cdot \cd{r} \passign \cd{x}
\cdot \cd{x} \passign \cd{n(x)} 
\cdot \passume (x = y)
\end{align*}
We have $\alpha(\rho_1) = \{ (\cd{key}(\cd{n}(\init{\cd{x}})),
\init{\cd{k}}) , (\cd{n}(\cd{n}(\init{\cd{x}})),\init{\cd{y}})\}$ and
$\beta(\rho_1) = \{ (\init{\cd{T}}, \init{\cd{F}}), (\init{\cd{x}},
\init{\cd{y}}), (\cd{key}(\init{\cd{x}}), \init{\cd{k}}),
(\cd{n}(\init{\cd{x}}), \init{\cd{y}})\}$.
\end{example}

\subsection*{Semantics of Programs}
We define the semantics of a program with respect to an algebra or
data model that gives interpretations to all the constants and
function symbols in the signature.  An execution $\rho$ is said to be
\emph{feasible} with respect to a data model if, informally, the set
of assumptions it makes are true in that model.  More precisely, for
an execution $\rho$, recall that $\alpha(\rho)$ and $\beta(\rho)$ are
the set of equality assumes and disequality assumes over terms computed in
$\rho$. An execution $\rho$ is feasible in a data-model $\Mm$, if for
every $(t,t') \in \alpha(\rho)$, $\termsem{t}{\Mm} =
\termsem{t'}{\Mm}$, and for every $(t,t') \in \beta(\rho)$,
$\termsem{t}{\Mm} \neq \termsem{t'}{\Mm}$.

\subsection*{The Verification Problem}

Let us now define the logic for postconditions, which are
quantifier-free formulas Boolean combination of equality
constraints on variables.
Given a finite set of variables $V$, the syntax for postconditions is
defined by the following logic $\Ll_=$.
\[
\Ll_=: ~~~~~\varphi ::=   x\!=\!y ~\mid~ \varphi \vee \varphi ~\mid~ \neg \varphi
\]
where above, $x,y\in V$.

Note that a more complex post-condition in the form of a quantifier-free formulae using the functions/relations/constants of the underlying data domain and the current variables can be 
incorporated by inserting code at the end of the program that computes the relevant terms, leaving the actual postcondition to check only properties of equality over variables.

\medskip
\noindent We can now define the verification problem for uninterpreted
programs.

\begin{definition}[The Verification Problem for Uninterpreted Programs]
Given a $\Ll_=$ formula $\varphi$ over a set of variables $V$, and a
program $s$ over $V$, determine, for every data-model $\Mm$ and
every execution $\rho \in \exec(s)$ that is feasible in $\Mm$, if
$\Mm$ satisfies the formula $\varphi$ under the interpretation that
maps every variable $x \in V$ to $\termsem{\comp(\rho,x)}{\Mm}$.\qed
\end{definition}

It is useful to observe that the verification problem for a program
$s$ with postcondition $\varphi$ in $\Ll_=$ can be reduced to the
verification of a program $s'$ with additional $\passume$s and
$\pif-\pthen-\pelse$ statements and postcondition $\fals$. 
Thus, without loss of
generality, we may assume that the postcondition is fixed to be
$\fals$. Observe that in this situation, the verification problem
essentially reduces to determining the existence of an execution that
is feasible in some data model. If there is a feasible execution (in
some data model) then the program violates its postcondition;
otherwise the program is correct.


\subsection{Main Results}
\seclabel{mainresults}

In this paper, we investigate the decidability of the verification
problem for uninterpreted programs. Our first result is that this
problem is, in general, undecidable.
We discuss this in detail in~\secref{undecidability}.

\medskip 
{\bf Result \#1:}
\emph{The verification problem for uninterpreted programs is undecidable.
}
\medskip

\noindent
Since the general verification problem is undecidable, we identify a
special class of programs for which the verification problem is
decidable. In order to describe what these special programs are, we
need to introduce a notion of \emph{coherent executions}. Observe that
as an execution proceeds, more (structurally) complex terms get
computed and assigned to variables, and more $\passume$ statements
identify constraints that narrow the collection of data models in
which the execution is feasible. Coherent executions satisfy two
properties. The first property, that we call \emph{memoizing}, requires
that if a term $t$ is computed (through an assignment) and either $t$
or something ``equivalent'' (w.r.t. to the equality assumes in the
execution) to $t$ was computed before in the execution, then it must
be currently stored in one of the variables. This is best illustrated
through an example. Consider the partial execution
\begin{align*}
\pi \delequal &
\, \passume (\code{x} \neq \code{z})
\cdot \code{y} \passign \code{n(x)}
\cdot \code{g} \passign \code{y}
\cdot \passume (\code{y} \neq \code{z})
\cdot \code{y} \passign \code{n(y)}
\cdot \passume (\code{y} \neq \code{z})
\cdot \code{x} \passign \code{n(x)}
\end{align*}
of program $P_3$ (in \figref{example1}). The term $\cd{n}(\init{\cd{x}})$ is
re-computed in the last step, but it is currently stored in the
variable $\code{g}$. On the other hand, a similar partial execution
\begin{align*}
\pi' \delequal &
\, \passume (\code{x} \neq \code{z})
\cdot \code{y} \passign \code{n(x)}
\cdot \passume (\code{y} \neq \code{z})
\cdot \code{y} \passign \code{n(y)}
\cdot \passume (\code{y} \neq \code{z})
\cdot \code{x} \passign \code{n(x)}
\end{align*}
of $P_2$ is not memoizing since when $\cd{n}(\init{\cd{x}})$ is
recomputed in the last step, it is not stored in any variable; the
contents of variable $\code{y}$, which stored $\cd{n}(\init{\cd{x}})$
when it was first computed, have been over-written at this point, or,
in other words, the term $\cd{n}(\init{\cd{x}})$ was ``dropped'' by
the execution before it was recomputed.  The second property that
coherent executions must satisfy is that any step of the form
$\dblqt{\passume(x = y)}$ in the execution comes ``early''. That is,
any superterms of the terms stored in $x$ and $y$ computed by the
execution upto this point, are still stored in the program variables
and have not been overwritten. The formal definition of coherent
executions will be presented later in~\secref{coherent-ver} . Finally,
a program is \emph{coherent} if all its executions are coherent. The
most technically involved result of this paper is that the
verification problem for coherent uninterpreted programs is decidable.

\medskip
{\bf Result \#2:}
\emph{
The verification problem for coherent uninterpreted programs is decidable.
}
\medskip

\noindent The notion of coherence is inspired by the notion of bounded
pathwidth, but is admittedly technical. However, we show that
determining if a given program is coherent is decidable; hence users of
the verification result need not ensure that the program are coherent
manually.

\medskip
{\bf Result \#3:}
\emph{
Given a program, the problem of checking whether it is coherent is decidable.
}
\medskip

\noindent
The notion of coherence has two properties, namely, that executions are memoizing and have early-assumes. Both these properties seem to be important for our decidability result. The verification problem for programs all of whose executions satisfy only one of these two conditions turns out to be undecidable.

 \medskip
 {\bf Result \#4:} \emph{The verification problem for uninterpreted
    programs whose executions are memoizing is undecidable.  The
   verification problem for uninterpreted programs whose executions
      have early assumes is undecidable.  }
 \medskip

\noindent
The memoizing and early-assume requirements of coherence may not be
satisfied by even simple programs. For example, program $P_2$ in \figref{example1} does not
satisfy the memoizing requirement as demonstrated by the partial execution
$\pi'$ above. However, many of these programs can be made coherent by
adding a finite number of \emph{ghost variables}. These ghost variables
are only written and never read, and therefore, play no role in the
actual computation. They merely remember terms that have been
previously computed and can help meet the memoizing and early-assume requirements. We show that given a budget of $k$ variables, we can \emph{automatically} check whether a corresponding coherent
program with $k$ additional ghost variables exists, and in fact
compute a regular automaton for its executions, and verify the
resulting coherent program. The notation and terminology for $k$-coherent programs is
more complex and we delay defining them formally to \secref{kcoherence}, where
they are considered.

\medskip
{\bf Result \#5:} \emph{ Given a program $P$ and $k \in \mathbb{N}$, we
  can decide whether executions of $P$ can be augmented with $k$ ghost
  variables and assignments so that they are coherent (i.e., check
  whether $P$ is $k$-coherent).  Furthermore, if such a coherent
  program exists, we can construct it and verify it against
  specifications.  }
\medskip

\noindent
Finally, in \secref{vpa}, we consider programs with recursive function
calls, and extend our results to them. In particular, we show the
following two results.

\medskip
{\bf Result \#6:} \emph{ The verification problem for coherent
  uninterpreted programs with recursive function calls is decidable.
}

\medskip
{\bf Result \#7} \emph{ Given a program $P$, with recursive function
  calls, and $k \in \mathbb{N}$, we can decide whether executions of
  $P$ can be augmented with $k$ local ghost variables (for each
  function) and interleaved ghost assignments that results in a
  coherent program. Furthermore, if such a coherent program exists, we
  can construct it and verify if against specifications.  }





\section{Verification of Coherent Uninterpreted Programs}
\seclabel{coherent-ver}

The verification problem for uninterpreted programs is undecidable; we
will establish this result in \secref{undecidability}. In this
section, we establish our main technical results, where we identify a
class of programs for which the verification problem is decidable. We
call this class of programs \emph{coherent}. We begin by formally
defining this class of programs. We then present our algorithm to
verify coherent programs. Finally, we conclude this section by showing
that the problem of determining if a given program is coherent is also
decidable.

Before presenting the main technical content of this section, let us
recall that an equivalence relation $\cong \subseteq \Terms \times
\Terms$ is said to be a \emph{congruence} if whenever $t_1 \cong
t_1'$, $t_2 \cong t_2'$, \ldots $t_m \cong t_m'$ and $f$ is an $m$-ary
function then $f(t_1,\ldots t_m) \cong f(t_1',\ldots t_m')$. Given a
binary relation $A \subseteq \Terms \times \Terms$, the
\emph{congruence closure} of $A$, denoted $\congcl{A}$, is the
smallest congruence containing $A$. 

For a congruence $\cong$ on $\Terms$, the equivalence class of a term
$t$ will be denoted by $\eqcl{t}{\cong}$; when $\cong = \congcl{A}$,
we will write this as $\eqcl{t}{A}$ instead of $\eqcl{t}{\congcl{A}}$.
For terms $t_1, t_2 \in \Terms$ and congruence $\cong$ on $\Terms$,
we say that \emph{$t_2$ is a superterm of $t_1$ modulo $\cong$}
if there are terms $t'_1, t'_2 \in \Terms$ such that 
$t'_1 \cong t_1$, $t'_2 \cong t_2$
and $t'_2$ is a superterm of $t'_1$.


\subsection{Coherent Programs}
\seclabel{coherence}

Coherence is a key property we exploit in our decidability results,
and is inspired by the concept of bounded pathwidth. In order to
define coherent programs we first need to define the notion of
coherence for executions. Recall that, for a partial execution $\rho$,
$\alpha(\rho)$ denotes the set of equality assumes made in $\rho$.

\begin{definition}[Coherent executions]
\deflabel{coherence}
We say that a (partial or complete) execution $\rho$ over variables $V$ is \emph{coherent} if 
it satisfies the following two properties.
\begin{description}
\item[Memoizing.] Let $\sigma' = \sigma \cdot \dblqt{x
  \!\passign\!\!f(\vec{z})}$ be a prefix of $\rho$ and let $t =
  \comp(\sigma' , x)$.  If there is a term $t' \in \Terms(\sigma)$
  such that $t' \congcl{\alpha(\sigma)} t$, then there must exist some
  $y \in V$ such that $\comp(\sigma, y) \congcl{\alpha(\sigma)} t$.
\item[Early Assumes.] 
Let $\sigma' = \sigma \cdot
  \dblqt{\passume(x\!=\!y)}$ be a prefix of $\rho$ and let $t_x =
  \comp(\sigma,x)$ and $t_y = \comp(\sigma,y)$.  
  If there is a term $t' \in \Terms(\sigma)$ such that
  $t'$ is either a superterm of $t_x$ or of $t_y$ modulo $\congcl{\alpha(\sigma)}$,
  then there must exist a variable $z\in V$ such that 
  $\comp(\sigma, z) \congcl{\alpha(\sigma)} t'$.
\end{description}
\end{definition}

\noindent
Formally, the memoizing property says that whenever a term $t$ is
recomputed (modulo the congruence enforced by the equality assumptions
until then), there must be a variable that currently corresponds to
$t$.  In the above definition, the assignment $x \passign f(\vec{z})$
is computing a term $t$, and if $t$ has already been computed (there
is a term $t'$ computed by the prefix $\sigma$ that is equivalent to
$t$), then we demand that there is a variable $y$ which after $\sigma$,
holds a term that is equivalent to $t$.

The second requirement of early assumes imposes constraints on when
$\dblqt{\passume(x=y)}$ steps are taken within the execution. We
require that such $\passume$ statements appear before 
the execution ``drops'' any computed term $t$ that 
is a superterm 
of the terms corresponding to $x$ and $y$,
i.e., before the execution reassigns the variables storing such 
superterms; notice that $\Terms(\sigma)$ also includes those terms that have been computed along the execution $\sigma$ and might have been dropped. 
Formally, we
demand that whenever an $\passume$ statement is executed equating
variables $x$ and $y$, if there is a superterm ($t'$) of either the term 
stored in $x$ or $y$ modulo the congruence so far, then there must be
a variable ($z$) storing a term equivalent to $t'$.

Finally, we come to the main concept of this section, namely, that of
coherent programs.

\begin{definition}[Coherent programs]
A coherent program is a program all of whose executions are coherent.
\end{definition}
 
\begin{example}
\exlabel{coherence} 
Consider the partial execution $\pi'$ of $P_2$ (\figref{example1})
that we considered in \secref{mainresults}.
\begin{align*}
\pi' \delequal &
\, \passume (\code{x} \neq \code{z})
\cdot \code{y} \passign \code{n(x)}
\cdot \passume (\code{y} \neq \code{z})
\cdot \code{y} \passign \code{n(y)}
\cdot \passume (\code{y} \neq \code{z})
\cdot \code{x} \passign \code{n(x)}
\end{align*}
Any extension of $\pi'$ to a complete execution of $P_2$, will not be
coherent. This is because $\rho'$ is not memoizing --- when
$\cd{n}(\init{\cd{x}})$ is recomputed in the last step, it is not stored in any
variable; the contents of variable $\code{y}$, which stored
$\cd{n}(\init{\cd{x}})$ when it was first computed, have been over-written at
this point.

On the other hand, the following execution over variables $\{\cd{x},
\cd{y}, \cd{z}\}$
\begin{align*}
\sigma \delequal &
\, \code{z} \passign \code{f(x)}
\cdot \code{z} \passign \code{f(z)}
\cdot \passume (\code{x} = \code{y})
\end{align*}
is also not coherent because $\dblqt{\passume (\code{x} = \code{y})}$
is not early. Observe that $\comp(\sigma, \cd{x}) = \init{\code{x}}$, 
$\comp(\sigma, \cd{y}) = \init{\cd{y}}$, and $\comp(\sigma, \cd{z}) =
\cd{f}(\cd{f}(\init{\cd{x}}))$. Now $\cd{f}(\init{\cd{x}}) \in
\Terms(\sigma)$, is a superterm of $\comp(\sigma,\cd{x})$ but is not
stored in any variable.

Consider the programs in \figref{example1}. $P_2$ is not coherent
because of partial execution $\pi'$ above. 
On the other hand, program $P_1$ is coherent.
This is because whenever an execution encounters 
$\dblqt{\passume (\cd{d} = \cd{k})}$,
both $\cd{d}$ and $\cd{k}$ have no superterms
computed in the execution seen so far.
The same holds for the $\dblqt{\passume( \cd{x} = \cd{y})}$
at the end of an execution due to the $\pwhile$ loop.
Further, whenever a term gets \emph{dropped}, or
over-written, it never gets computed again, even modulo
the congruence induced by the assume equations.
Similar reasoning establishes that $P_3$ is also coherent.
\end{example}


\subsection{Verifying Coherent Programs}
\seclabel{verifying-coherent-programs}

We are now ready to prove that the verification problem for coherent
programs is decidable. Recall that, without loss of generality, we may
assume that the postcondition is $\fals$\mpcomment{Removing as example is gone now: (see
\exref{trans-postcondition})}. Observe that, when the postcondition is
$\fals$, a program violates the postcondtion, if there is an execution
$\rho$ and a data model $\Mm$ such that $\rho$ is feasible in $\Mm$,
i.e., every equality and disequality assumption of $\rho$ holds in
$\Mm$. On the face of it, this seems to require evaluating executions
in all possible data models. But in fact, one needs to consider only
one class of data models. We begin by recalling the notion of an
initial model.

Given a binary relation $A \subseteq \Terms \times \Terms$ of
equalities, $\Mm$ is said to \emph{satisfy} $A$ (or $A$ \emph{holds in}
$\Mm$) if for every pair $(t,t') \in A$, $\termsem{t}{\Mm} =
\termsem{t'}{\Mm}$. For a relation $A$, there is a canonical model in
which $A$ holds.
\begin{definition}
\deflabel{term-model} 

The \emph{initial term model} for $A \subseteq \Terms \times \Terms$
over an algebraic signature $\Sigma = (\Cc, \Ff)$ is $\termmod{A} = (U,
\setpred{\sem{c}}{c \in \Cc}, \setpred{\sem{f}}{f \in \Ff})$ where
\begin{itemize}
\item $U = \Terms/\congcl{A}$,
\item $\sem{c} = \eqcl{c}{A}$ for any $c \in \Cc$, and
\item $\sem{f}(\eqcl{t_1}{A}, \ldots \eqcl{t_m}{A}) =
  \eqcl{f(t_1,\ldots t_m)}{A}$ for any $m$-ary function symbol $f \in
  \Ff$ and terms $t_1,\ldots t_m \in \Terms$. 
\end{itemize}
\end{definition}
An important property of the initial term model is the following.
\begin{proposition}
\lemlabel{term-model}
Let $A$ be a binary relation on terms, and $\Mm$ be any model
satisfying $A$. For any pair of terms $t,t'$, if
$\termsem{t}{\termmod{A}} = \termsem{t'}{\termmod{A}}$ then
$\termsem{t}{\Mm} = \termsem{t'}{\Mm}$.
\end{proposition}
\begin{proof}
Any model $\Mm$ defines an equivalence on terms $\equiv_{\Mm}$ as
follows: $t \equiv_{\Mm} t'$ iff $\termsem{t}{\Mm} =
\termsem{t'}{\Mm}$. Observe that $\equiv_{\Mm}$ is a congruence, and
if $\Mm$ satisfies $A$, then $A \subseteq \equiv_{\Mm}$. Thus,
$\congcl{A} \subseteq \equiv_{\Mm}$. Next, observe that for the term
model $\termmod{A}$, $\equiv_{\termmod{A}} = \congcl{A}$. The
proposition follows from these observations.
\end{proof}

One consequence of the above proposition
is the following. Let $A$ be a
set of equalities, $t_1, t_2$ be terms, and $\Mm$ be a data model
satisfying $A$. If $\termsem{t_1}{\Mm} \neq \termsem{t_2}{\Mm}$ then
$\termsem{t_1}{\termmod{A}} \neq \termsem{t_2}{\termmod{A}}$. This
means that to check the feasibility of an execution $\rho$, it
suffices to check its feasibility in $\termmod{\alpha(\rho)}$.
\begin{corollary}
\label{cor:feasibility}
Let $\rho$ be any execution. There is a data model $\Mm$ such that
$\rho$ is feasible in $\Mm$ if and only if $\rho$ is feasible in
$\termmod{\alpha(\rho)}$.
\end{corollary}
\noindent
That is, an execution $\rho$ is feasible iff $\congcl{\alpha(\rho)} \cap \beta(\rho) = \emptyset$.
%

Let us return to the problem of verifying if a program satisfies the
postcondition $\fals$. This requires us to check that no execution of
the program is feasible in any data model. Let us now focus on the
simpler problem of \emph{execution verification} --- given an
execution $\rho$ check if there is some data model in which $\rho$ is
feasible. If we can solve the execution verification problem, then we
could potentially solve the program verification problem; since the
set of executions of a program are regular, we could run the execution
verification algorithm synchronously with the NFA representing the set
of all program executions to see if any of them are feasible.

Corollary~\ref{cor:feasibility} has an important consequence for
execution verification --- to check if $\rho$ is feasible, evaluate
$\rho$ in the data model $\termmod{\alpha(\rho)}$. If the execution
verification algorithm is to be lifted to verify all executions of a
program, then the algorithm must evaluate the execution as the symbols
come in. It cannot assume to have the entire execution. This poses
challenges that must be overcome. First the term model
$\termmod{\alpha(\rho)}$ is typically infinite and cannot be
constructed explicitly. Second, since equality assumptions come in as
the execution unfolds, $\alpha(\rho)$ is not known at the beginning
and therefore, neither is the exact term model on which $\rho$ must be
evaluated known. In fact, in general, we cannot evaluate an arbitrary
execution $\rho$ in a term model $\termmod{\alpha(\rho)}$ in an
incremental fashion. The main result of this section shows that we can
exploit properties of coherent executions to overcome these
challenges.

To explain the intuition behind our algorithm, let us begin by
considering a na\"{i}ve algorithm that evaluates an execution in a
data model. Suppose the data model $\Mm$ is completely known. One
algorithm to evaluate an execution $\rho$ in $\Mm$, would keep track
of the values of each program variable with respect to model $\Mm$,
and for each $\passume$ step, check if the equality or disequality
assumption holds in $\Mm$. When $\Mm$ is the term model, the value
that variable $x$ takes after (partial) execution $\rho'$, is the
equivalence class of $\comp(\rho',x)$ with respect to congruence
defined by all the equality assumptions in the \emph{complete} execution $\rho$.

Our algorithm to verify an execution, will follow the basic template
of the na\"{i}ve algorithm above, with important modifications. First,
after a prefix $\rho'$ of the execution $\rho$, we have only seen a
partial set of equality assumptions and not the entire set. Therefore,
the value of variable $x$ that the algorithm tracks will be
$\eqcl{\comp(\rho',x)}{\alpha(\rho')}$ and not
$\eqcl{\comp(\rho',x)}{\alpha(\rho)}$. Now, when a new equality
assumption $\passume (y = z)$ is seen, we will need to update the
values of each variable to be that in the term model that also
satisfies this new equation. This requires updating the congruence
class of the terms corresponding to each variable as new equalities
come in. In addition, the algorithm needs to ensure that if a previously
seen disequality assumption is violated because of the new equation,
it can be determined, eventhough the disequality assumption maybe
between two terms that are no longer stored in any of the program
variables. Second, our algorithm will also track the interpretation of
the function symbols when applied to the values stored for variables
in the program. Thus, after a prefix $\rho'$, the algorithm constructs
part of the model $\termmod{\alpha(\rho')}$ when restricted to the
variable values. This partial model helps the algorithm update the
values of variables when a new equality assume is read. The third
wrinkle concerns how $\eqcl{\comp(\rho',x)}{\alpha(\rho')}$ is
stored. We could store a representative term from this equivalence
class. This would result in an algorithm whose memory requirements
grow with the execution. Instead the algorithm only maintains, for
every pair of variables $x, y$, whether their values in
$\termmod{\alpha(\rho')}$ are equal or not. This means that the memory
requirements of the algorithm do not grow with the length of the
execution being analyzed. Thus, we will in fact show, that the
collection of all feasible partial executions is a regular language.

In order to be able to carry out the above analysis incrementally, our
algorithm crucially exploits the coherence properties of the
execution. To illustrate one reason why the above approach would not
work for non-coherent executions, consider a prefix $\rho'$ such that
$\comp(\rho',x) = \init{x}$, $\comp(\rho',y) = \init{y}$,
$\comp(\rho',u) = f^{100}(\init{x})$, and $\comp(\rho',v) =
f^{100}(\init{y})$. Let us assume that $\alpha(\rho') =
\emptyset$. Suppose we now encounter $\passume (x = y)$. This means
that in the term model that satisfies this equality, the values of
variables $u$ and $v$ are the same. However, this is possible only if
the algorithm somehow maintains the information that $u$ and $v$ are
the result of hundred applications of $f$ to the values in $x$ and
$y$. This cannot be done using bounded memory. Notice, however, that
in this case $\rho'\cdot \dblqt{\passume (x = y)}$ is not a coherent
execution because the assume at the end is not early. Early assumes
ensure that the effect of an new equality assumption can be fully
determined on the current values to the variables.

  To understand the importance of the memoizing property in the decision procedure,
  consider the execution 
  $\rho' \delequal x \passign y \cdot \underbrace{y \passign f(y) \cdots y \passign f(y)}_{\text{n times}} \cdot \underbrace{x \passign f(x) \cdots x \passign f(x)}_{\text{n times}} $.
  This execution trivially satisfies the ``early assumes'' criterion.
  However it is not memoizing since the terms 
  $\init{y}, f(\init{y}), \ldots, f^{n-1}(\init{y})$ 
  have been re-computed after they have been dropped.
  Now, suppose that $\rho \delequal \rho' \cdot \passume(x \neq y)$ 
  is a complete extension of $\rho'$.
  Notice that $\rho$ is not memoizing but still satisfies the ``early assumes'' criterion 
  (as it has no equality assumptions).
  Now, in order for the algorithm to correctly determine that this execution is infeasible,
  it needs to correctly maintain the information that $\comp(\rho',y) = \comp(\rho',x) = f^{n}(\init{y})$.
  This again, is not possible using bounded memory.

\subsection*{Formal Details}

We will now flesh out the intuitions laid out above. We will introduce
concepts and properties that will be used in the formal construction
and its correctness proof.

Recall that our algorithm will track the values of the program
variables in a term model. 
When we have a coherent execution $\sigma' = \sigma \cdot \passume(x = y)$,
the terms corresponding to program variables 
obey a special relationship with the 
set of terms $\Terms(\sigma')$ constructed anytime during the execution 
and with the equality assumptions seen in $\sigma$. 
We capture this through the following definition
which has been motivated by the condition of early assumes.

\begin{definition}[Superterm closedness modulo congruence]
\deflabel{superterm-closure}
Let $T$ be a subterm closed set of terms.  
Let $E \subseteq T\times T$
be a set of equations on $T$ and $\congcl{E}$ be its congruence closure.
Let $W \subseteq T$ and let $t_1, t_2 \in W$.
Then, $W$ is said to be \emph{closed under
  superterms with respect to $T$, $E$ and $(t_1, t_2)$} 
if for any term $s \in T$ such that
$s$ is a superterm of either $t_1$ or $t_2$ modulo $\congcl{E}$,
there is term  $s' \in W$ such that $s' \congcl{E} s$.

\end{definition}
Coherent executions ensure that the set of
values of variables (or equivalently the set of terms corresponding to the variables) 
is superterm closed modulo congruence with respect to a newly encountered equality assumption; 
observe that for any partial execution $\rho$, $\Terms(\rho)$ is subterm closed.
\begin{lemma}
\lemlabel{coherent-superterm}
Let $\sigma$ be a coherent execution over variables $V$ and let $\rho' = \rho \cdot \dblqt{\passume(x = y)}$ be any prefix of $\sigma$. 
Then $W = \setpred{\comp(\rho,v)}{v \in V}$
is closed under superterms with respect to $\Terms(\rho)$,
$\alpha(\rho)$ and $(\comp(\rho, x), \comp(\rho, y))$.
\end{lemma}

As pointed out in the overview, our algorithm will not explicitly
track the terms stored in program variable, but instead track the
equivalence between these terms in the term model. In addition, it
also tracks the interpretations of function symbols on the stored
terms in the term model. Finally, it will store the pairs of terms
(stored currently in program variables) that have been assumed to be
not equal in the execution. The following definition captures when
such an algorithm state is consistent with a set of terms, equalities,
and disequalities. 
In the definition below, the reader may
think of $W$ as the set of terms corresponding to each program variable, $E$
as the set equality assumptions, and $D$ as the set of disequality
assumptions after a prefix of the execution.

\begin{definition}[Consistency]
\deflabel{consistency}
Let $W$ be a set of terms, $E$ a set of equations on terms, and $D$ be
a set of disequalities on terms. Let $\equiv_W$ be an equivalence
relation on $W$, $D_W \subseteq W/\equiv_{W} \times W/\equiv_{W}$ be a symmetric relation, 
and $P$ be a partial interpretation of function symbols, i.e., for any
$k$-ary function symbol $f$, $P(f)$ is a partial function mapping
$k$-tuples in $(W/\equiv_{W})^k$ to $W/\equiv_W$. We will say
$(\equiv_W,D_W,P)$ is \emph{consistent} with respect to $(W,E,D)$ iff
the following hold.
\begin{enumerate}[label=(\alph*)]
\item 
For $t_1,t_2 \in W$, $t_1 \equiv_W t_2$ if and only if
$\termsem{t_1}{\termmod{E}} = \termsem{t_2}{\termmod{E}}$, i.e.,
$t_1$ and $t_2$ evaluate to the same value in $\termmod{E}$,

\item 
$(\eqcl{t_1}{\equiv_W},\eqcl{t_2}{\equiv_W}) \in D_W$ iff there
are terms $t_1',t_2'$ such that $t_1'\congcl{E} t_1$, and $t_2' \congcl{E} t_2$
and $\set{(t_1',t_2'), (t'_2, t'_1)} \cap D \neq \emptyset$. 
\item

$P(f)(\eqcl{t_1}{\equiv_W},\ldots , \eqcl{t_k}{\equiv_W}) = 
   \begin{cases}
   \eqcl{t}{\equiv_W} & \mbox{if } f(t_1,\ldots t_k) \congcl{E} t\\
   \undf & \mbox{otherwise}
   \end{cases}$
\end{enumerate}
\end{definition}

There are two crucial properties about a set $W$ that is superterm
closed (\defref{superterm-closure}). When we have a state that is
consistent (as per \defref{consistency}), we can correctly update it
when we add an equation by doing a ``local'' congruence closure of the
terms in $W$. This is the content of \lemref{window_update_equation}
and its detailed proof can be found in Appendix~\ref{app:proofs-automaton}.

\begin{lemma}
\lemlabel{window_update_equation}

Let $T$ be a set of subterm-closed set of terms, $E \subseteq T \times
T$ be a set of equalities on $T$, and $D \subseteq T \times T$ be a
set of disequalities. Let $W \subseteq T$ be a set closed under superterms
with respect to $T, E$ and some pair $(s_1, s_2) \in W\times W$. 
Let $(\equiv_W, D_W, P)$ be consistent
with $(W,E,D)$. 
Define $\sim_{s,s'}$ to be the smallest equivalence relation on $W$ such that
\begin{itemize}
\item $\equiv_W \cup \{(s,s')\} \subseteq \sim_{s,s'}$
\item for every $k$-ary function symbol $f$ and terms
  $t_1,t_1',t_2,t_2',\ldots t_k,t_k',t,t' \in W$ such that $t \in
  P(f)(\eqcl{t_1}{\equiv_W}, \ldots \eqcl{t_k}{\equiv_W})$, $t' \in
  P(f)(\eqcl{t_1'}{\equiv_W},\ldots \eqcl{t_k'}{\equiv_W})$, and
  $(t_i,t_i') \in \sim_{s,s'}$ for each $i$, we have $(t,t') \in
  \sim_{s,s'}$.
\end{itemize}
In addition, take $D'_W =
\setpred{(\eqcl{t_1}{\sim_{s,s'}},\eqcl{t_2}{\sim_{s,s'}})}{(\eqcl{t_1}{\equiv_W},\eqcl{t_2}{\equiv_W})
  \in D_W}$ and
\[
P'(f)(\eqcl{t_1}{\sim_{s,s'}},\ldots \eqcl{t_k}{\sim_{s,s'}}) =
  \left\{
  \begin{array}{ll}
  \eqcl{t}{\sim_{s,s'}} & \mbox{if } 
                 P(f)(\eqcl{t_1}{\equiv_W},\ldots \eqcl{t_k}{\equiv_W}) = 
                 \eqcl{t}{\equiv_W} \\
  \undf & \mbox{otherwise}
  \end{array} \right.
\]
Then $(\sim_{s,s'},D'_W,P')$ is consistent with $(W, E\cup\{(s,s')\}, D)$.
\end{lemma}

The second important property about $W$ being superterm closed is that
feasibility of executions can be checked easily. Recall that, our
previous observations indicate that an execution is feasible, if all
the disequality assumptions hold in the term model, i.e., if $E$ is a
set of equality assumptions and $D$ is a set of disequality assumptions,
feasibility requires checking that $\congcl{E} \cap D =
\emptyset$. Now, we show that when $W$ is superterm closed, 
then checking this condition when a new equation is added to $E$
can be done by just looking at $W$; notice
that $D$ may have disequalities involving terms that are not in $W$,
and so the observation is not trivial.

\begin{lemma}
\lemlabel{window_check_eq}

Let $T$ be a set of subterm-closed set of terms, $E \subseteq T \times
T$ be a set of equalities on $T$, and $D \subseteq T \times T$ be a
set of disequalities such that $D \cap \congcl{E} = \emptyset$. Let $W
\subseteq T$ and let $t_1, t_2 \in W$ be such that
$W$ is closed under superterms with respect to $T, E$ and $(t_1, t_2)$.
Let $(\equiv_W, D_W, P)$ be consistent with $(W,E,D)$. Then, 
$\congcl{E \cup \{(t_1,t_2)\}} \cap D \neq \emptyset$ iff there
  are terms $t_1',t_2' \in W$ such that $(\eqcl{t_1'}{\equiv_W},
  \eqcl{t_2'}{\equiv_W}) \in D_W$ and $t_1' \sim_{t_1,t_2} t_2'$,
  where $\sim_{t_1,t_2}$ is the equivalence relation on $W$ defined in
  \lemref{window_update_equation}.
\end{lemma}

Further, the notion of consistency allows us to correctly
check for feasibility when a disequality assumption is seen, and
this is formalized below.

\begin{lemma}
\lemlabel{window_check_diseq}
Let $T$ be a set of subterm-closed set of terms,
$W \subseteq T$, $E \subseteq T \times
T$ be a set of equalities on $T$, and $D \subseteq T \times T$ be a
set of disequalities such that $D \cap \congcl{E} = \emptyset$. 
Let $(\equiv_W, D_W, P)$ be consistent with $(W,E,D)$ and let $t_1, t_2 \in W$. 
Then, 
$\congcl{E} \cap (D \cup \{(t_1,t_2)\} \neq \emptyset$ iff
$(t_1,t_2)\in \equiv_W$.
\end{lemma}

Lemmas~\ref{lem:coherent-superterm}, \ref{lem:window_update_equation},
\ref{lem:window_check_eq} and~\ref{lem:window_check_diseq} suggest that when an execution is coherent,
the equivalence between terms stored in program variables can be
tracked and feasibility of the execution can be checked, as equality
and disequality assumptions are seen. 
Next, we will exploit these observations to give the construction of an 
automaton that accepts exactly those coherent executions that are feasible.

\subsection*{Streaming Congruence Closure: Automaton for Feasibility of Coherent Executions}

Having given a broad overview of our approach, and defined various
concepts and properties, we are ready to present the main result of this
section, which says that the collection of feasible partial executions
forms a regular language. Let us fix $V$ to be the set of program
variables and $\Sigma = (\Cc,\Ff)$ to be the signature of
operations. Recall that $\Pi$ denotes the alphabet over which
executions are defined. We now formally define the automaton
$\feas{\Aa}$ whose language is the collection of all partial executions
that are feasible in some data model.

\vspace*{0.1in}
\noindent
\textbf{States.}  The states in our automaton are either the special
state $\reject$ or tuples of the form $(\equiv,d,P)$ where
$\equiv \subseteq V \times V$ is an equivalence relation, $d \subseteq
V/\equiv \times V/\equiv$ is symmetric and \emph{irreflexive}, and $P$ be a partial
interpretation of the function symbols in $\Ff$, i.e., for any $k$-ary
function symbol $f$, $P(f)$ is a partial function from $(V/\equiv)^k$
to $V/\equiv$. Intuitively, in a state of the form $(\equiv,d,P)$,
$\equiv$ captures the equivalence amongst the terms stored in the
variables that hold in the term model satisfying the equalities seen
so far, $d$ are the disequality assumptions seen so far restricted to
the terms stored in the program variables, and $P$ summarizes the
interpretations of the function symbols in the relevant term model.

\vspace*{0.1in}
\noindent
\textbf{Initial state.} The initial state $q_0 = (\equiv_0, d_0,
P_0)$, where $\equiv_0 = \setpred{(x,x)}{x \in V}$, $d_0 = \emptyset$,
and for every $k$-ary function symbol $f$, and any $x_1,\ldots x_k \in
V$, $P(f)(\eqcl{x_1}{\equiv_0}, \ldots \eqcl{x_k}{\equiv_0}) = \undf$ (undefined).

\vspace*{0.1in}
\noindent
\textbf{Accepting states.} All states except $\reject$ are
accepting.

\vspace*{0.1in}
\noindent
\textbf{Transitions.} The $\reject$ state is absorbing, i.e.,
for every $a \in \Pi$, $\feas{\delta}(\reject,a) =
\reject$. For the other states, the transitions are more
involved. Let $q = (\equiv,d,P)$ and let $q' = (\equiv',d',P')$. There
is a transtion $\feas{\delta}(q,a) = q'$ if one of the following
conditions holds. If in any of the cases $d'$ is \emph{not
  irreflexive} then $\feas{\delta}(q,a) = \reject$ in each case.
For a subset $V' \subseteq V$, we will use the notation 
$\longproj{\equiv}{V'}$ to denote the relation $\big( \equiv \cap V' \times V'\big)$.
\begin{description}
\item[$a = \dblqt{x \passign y}$].\\
In this case if $y$ and $x$ are the same variables, then $\equiv' \, = \, \equiv$, $d' = d$ and $P' = P$.
Otherwise, the variable $x$ gets updated to be in the equivalence class of the variable $y$, and $d'$ and
  $P'$ are updated in the most natural way. Formally,
\begin{itemize}
\item $\equiv' =  
  \longproj{\equiv}{V\setminus\set{x}}
  \cup \setpred{(x,y'),(y',x)}{y'\equiv
  y} \cup \set{(x, x)}$.
\item $d' = \setpred{(\eqcl{x_1}{\equiv'},\eqcl{x_2}{\equiv'})}{x_1,
  x_2 \in V\setminus \set{x}, (\eqcl{x_1}{\equiv},\eqcl{x_2}{\equiv})
  \in d }$ 
\item $P'$ is such that for every $r$-ary function $h$,
\[
P'(h)(\eqcl{x_1}{\equiv'},\ldots \eqcl{x_r}{\equiv'}) =
   \begin{cases}
   \eqcl{u}{\equiv'} 
                & x \not\in \set{u, x_1,\ldots x_r} \text{ and }\\
   & \eqcl{u}{\equiv} = P(h)(\eqcl{x_1}{\equiv},\ldots \eqcl{x_r}{\equiv})\\
   \undf & \mbox{otherwise}
   \end{cases}
\]
\end{itemize}

\item[$a = \dblqt{x \passign f(z_1, \ldots z_k)}$].\\
 There are two cases  to consider.
\begin{enumerate}
\item \textbf{Case $P(f)(\eqcl{z_1}{\equiv},\ldots \eqcl{z_k}{\equiv})$
  is defined}.\\
   Let $P(f)(\eqcl{z_1}{\equiv},\ldots \eqcl{z_k}{\equiv})
  = \eqcl{v}{\equiv}$. 
  This case is similar to the case when $a$ is
  $\dblqt{x \passign y}$.  
  That is, when $x \in \eqcl{v}{\equiv}$, then $\equiv' \, = \, \equiv$, $d' = d$ and $P' = P$.
  Otherwise, we have
\begin{itemize}
\item $\equiv' = \longproj{\equiv}{V\setminus\set{x}} \cup \setpred{(x,v'),(v',x)}{v'\equiv v} \cup \set{(x, x)}$
\item $d' = \setpred{(\eqcl{x_1}{\equiv'},\eqcl{x_2}{\equiv'})}{x_1,
  x_2 \in V\setminus \set{x}, (\eqcl{x_1}{\equiv},\eqcl{x_2}{\equiv})
  \in d }$
\item $P'$ is such that for every $r$-ary function $h$,
\[
P'(h)(\eqcl{x_1}{\equiv'},\ldots \eqcl{x_r}{\equiv'}) =
   \begin{cases}
   \eqcl{u}{\equiv'} 
   
   & x \not\in \set{u, x_1,\ldots x_r} \text{ and }\\
   & \eqcl{u}{\equiv} = P(h)(\eqcl{x_1}{\equiv},\ldots \eqcl{x_r}{\equiv})\\
   \undf & \mbox{otherwise}
   \end{cases}
\]
\end{itemize}

\item \textbf{Case $P(f)(\eqcl{z_1}{\equiv},\ldots \eqcl{z_k}{\equiv}$
  is undefined.} \\
  In this case, we remove $x$ from its older
  equivalence class and make a new class that only contains the
  variable $x$.  We update $P$ to $P'$ so that the function $f$ maps
  the tuple $(\eqcl{z_1}{\equiv'}, \ldots, \eqcl{z_k}{\equiv'})$ (if
  each of them is a valid/non-empty equivalence class) to the class
  $\eqcl{x}{\equiv'}$.  The set $d'$ follows easily from the new
  $\equiv'$ and the older set $d$.  Thus,
\begin{itemize}
\item $\equiv' = \longproj{\equiv}{V\setminus\set{x}} \cup
  \set{(x, x)}$
\item $d' = \setpred{(\eqcl{x_1}{\equiv'},\eqcl{x_2}{\equiv'})}{x_1,
  x_2 \in V\setminus \set{x}, (\eqcl{x_1}{\equiv},\eqcl{x_2}{\equiv})
  \in d }$
\item $P'$ behaves similar to $P$ for every function different from $f$.
\begin{itemize}
\item For every $r$-ary function $h \neq f$,
\[
P'(h)(\eqcl{x_1}{\equiv'}, \ldots, \eqcl{x_r}{\equiv'}) = 
   \begin{cases}
   \eqcl{u}{\equiv'} & \mbox{if } x \not\in \set{u, x_1,\ldots x_k}
        \mbox{ and } \\
        & \eqcl{u}{\equiv} = P(h)(\eqcl{x_1}{\equiv},\ldots \eqcl{x_r}{\equiv})\\
   \undf & \mbox{otherwise}
   \end{cases}
\]
\item For the function $f$, we have the following.
\[
P'(f)(\eqcl{x_1}{\equiv'}, \ldots, \eqcl{x_k}{\equiv'}) =
   \begin{cases}
   \eqcl{x}{\equiv'} & \mbox{if } x_i = z_i\ \forall i \text{ and } x \not\in \set{x_1,\ldots x_k}\\
   \eqcl{u}{\equiv'} & \mbox{if } x \not\in \set{u, x_1,\ldots x_k}
        \mbox{ and } \\
        & \eqcl{u}{\equiv} = P(f)(\eqcl{x_1}{\equiv},\ldots \eqcl{x_
k}{\equiv})\\
   \undf & \mbox{otherwise}
   \end{cases}
\]
\end{itemize}            
\end{itemize}
\end{enumerate}

\item[$a = \dblqt{\passume(x = y)}$].\\
Here, we essentially merge the
  equivalence classes in which $x$ and $y$ belong and perform the ``local
  congruence closure'' (as in \lemref{window_update_equation}). In
  addition, $d'$ and $P'$ are also updated as in
  \lemref{window_update_equation}. 	
\begin{itemize}
\item $\equiv'$ is the smallest equivalence relation on $V$ such that
  (a) $\equiv \cup \set{(x,y)} \subseteq \equiv'$, and (b) for every
  $k$-ary function symbol $f$ and variables $x_1,x_1',x_2,x_2',\ldots
  x_k,x_k',z,z' \in V$ such that $\eqcl{z}{\equiv} = P(f)(\eqcl{x_1}{\equiv},
  \ldots \eqcl{x_k}{\equiv})$, $\eqcl{z'}{\equiv} =
  P(f)(\eqcl{x_1'}{\equiv},\ldots \eqcl{x_k'}{\equiv})$, and
  $(x_i,x_i') \in \equiv'$ for each $i$, we have $(z,z') \in
  \equiv'$.
\item $d' =
  \setpred{(\eqcl{x_1}{\equiv'},\eqcl{x_2}{\equiv'})}{(\eqcl{x_1}{\equiv},\eqcl{x_2}{\equiv})
    \in d }$
\item $P'$ is such that for every $r$-ary function $h$,
\[
P'(h)(\eqcl{x_1}{\equiv'}, \ldots \eqcl{x_r}{\equiv'}) = 
   \begin{cases}
   \eqcl{u}{\equiv'} & \mbox{if } \eqcl{u}{\equiv} = P(h)(\eqcl{x_1}{\equiv},\ldots \eqcl{x_r}{\equiv}) \\
   \undf & \mbox{otherwise}
   \end{cases} 
\]
\end{itemize}

\item[$a = \dblqt{\passume(x \neq y)}$].\\
In this case,
\begin{itemize}
\item $\equiv' = \equiv$
\item $d' = d \cup \set{(\eqcl{x}{\equiv'},\eqcl{y}{\equiv'}), (\eqcl{y}{\equiv'},\eqcl{x}{\equiv'})}$
\item $P' = P$
\end{itemize}
\end{description}

The formal description of automaton $\feas{\Aa}$ is complete. We begin
formally stating the invariant maintained by the automaton during its
computation.

\begin{lemma}
\lemlabel{inv_automaton} Let $\rho$ be a coherent partial
execution. Let $q_\rho$ be the state reached by $\feas{\Aa}$ after
reading $\rho$. The following properties are true.
\begin{enumerate}
\item If $\rho$ is infeasible then $q_\rho = \reject$.
\item If $\rho$ is feasible in some data model, then $q_\rho$ is of
  the form $(\equiv,d,P)$ such that $(\equiv,d,P)$ is consistent with
  $(\setpred{\comp(\rho,x)}{x \in V}, \alpha(\rho), \beta(\rho))$; see
  \defref{consistency} for the notion of consistency.
\end{enumerate}
\end{lemma}




Assuming that the signature $\Sigma$ is of constant size, 
the automaton $\feas{\Aa}$ has $O(2^{|V|^{O(1)}})$ states.

\subsection*{Verifying a Coherent Program}

We have so far described a finite memory, streaming algorithm that 
given a coherent execution can determine if it is feasible in any data
model by computing congruence closure. We can use that algorithm to verify coherent, uninterpreted
programs. 

\begin{theorem}
\thmlabel{verifying-coherent}
Given a coherent program $s$ with postcondition $\fals$, the problem
of verifying $s$ is $\pspc$-complete.
\end{theorem}
\begin{proof}[Proof Sketch]
Observe that a coherent program with postcondition $\fals$ is correct
if it has no feasible executions. If $\Aa_s$ is the NFA accepting
precisely the executions of $s$, the goal is to determine if $L(\Aa_s)
\cap L(\feas{\Aa}) \neq \emptyset$. This can be done taking the cross
product of the two automata and searching for an accepting
path. Notice that storing the states of $\feas{\Aa}$ uses space that is
polynomial in the number of variables, and the cross product automaton
can be constructed on the fly. This gives us a $\pspc$ upper bound for
the verification problem.

The lower bound is obtained through a reduction from Boolean program
verification. Recall that Boolean programs are imperative programs
with while loops and conditionals, where all program variables take on
Boolean values. The verification problem for such programs is to
determine if there is an execution of the program that reaches a
special $\codekey{halt}$ statement. This problem is known to be
$\pspc$-hard. The Boolean program verification problem can be reduced
to the verification problem for uninterpreted programs --- the
uninterpreted program corresponding to a Boolean program will have no
function symbols in its signature, have constants for true and false,
and have two program variables that are never modified which store
true and false respectively. Since in such a program the variables
never store terms other than constants, the executions are trivially
coherent.
\end{proof}


\subsection{Decidability of Coherence}
\seclabel{dec-coherence}

Manually checking if a program is coherent is difficult. However, in
this section we prove that, given a program $s$, checking if $s$ is
coherent can be done in $\pspc$. The crux of this result is an
observation that the collection of coherent executions form a regular
language. This is the first result we will prove, and we will use this
observation to give a decision procedure for checking if a program is
coherent. Let us recall that executions of programs over variables $V$
are words over the alphabet $\Pi$.

\begin{theorem}
\thmlabel{coherent-reg}
The language $\coh{L} = \setpred{\rho \in \Pi^*}{\rho \mbox{ is
    coherent}}$ is regular. More precisely, there is a DFA $\coh{\Aa}$
of size $O(2^{V^{O(1)}})$ such that $L(\coh{\Aa}) = \coh{L}$.
\end{theorem}

\begin{proof}
Observe that an execution $\rho$ is coherent if and only if the
execution $\proj{\rho}{\Pi\setminus\setpred{\passume (x \neq y)}{x,y
    \in V}}$ is coherent; here
$\proj{\rho}{\Pi\setminus\setpred{\passume (x \neq y)}{x,y \in V}}$ is
the execution obtained by dropping all the disequality assumes from
$\rho$. Hence, the automaton $\coh{\Aa}$ will ignore all the
disequality assumes and only process the other steps.

The automaton $\coh{\Aa}$ heavily uses the automaton $\feas{\Aa}$,
constructed in \secref{verifying-coherent-programs}. Intuitively,
given a word $\rho$, our algorithm inductively checks whether prefixes
of $\rho$ are coherent. Hence, when examining a prefix $\sigma \cdot
a$, we can assume that $\sigma$ is coherent and use the properties of
the state of $\feas{\Aa}$ obtained on $\sigma$.

The broad outline of how $\coh{\Aa}$ works is as follows.
\begin{itemize}
\item We keep track of the state of the automaton $\feas{\Aa}$. Recall
  that its state is of the form $(\equiv, d, P)$, where $\equiv$
  defines an equivalence relation over the variables $V$, $d$ is a set
  of disequalities, and $P$ is a partial interpretation of the
  function symbols in the signature, restricted to contents of the
  program variables. Now, since we will drop all the disequality
  assumes, $d$ in this context is always $\emptyset$ and so the state
  of $\feas{\Aa}$ will never be $\reject$.
 
\item In addition, the state of $\coh{\Aa}$ will have a function $E$
  that associates with each function symbol $f$ of arity $k$, a
  function from $(V/\equiv)^k$ to $\set{\tru,\fals}$. Intuitively, $E$
  tracks, for each function symbol $f$ and each tuple of variables,
  whether $f$ has ever been computed on that tuple at any point in the
  execution --- $E(f)(\vec{z}) = \tru$, if $f(\vec{z})$ has been
  computed, and is $\fals$ if it has not. Note that, if
  $P(f)(\vec{z})$ is defined (here $P$ is the component that is part
  of the state of $\feas{\Aa}$) then $E(f)(\vec{z})$ will definitely
  be $\tru$. The role of $E(f)$, however, is to remember in addition
  whether $f$ has been computed on terms but the image has been
  ``dropped'' by the program.
\end{itemize}

The update of the extra information $E(f)$ is not hard. Whenever
$\coh{\Aa}$ reads $\dblqt{x \passign f(z_1,\ldots z_k)}$, we set
$E(f)(\eqcl{z_1}{\equiv},\ldots \eqcl{z_k}{\equiv})$ to $\tru$. If
$E(f)(\eqcl{z_1}{\equiv},\ldots \eqcl{z_k}{\equiv}) = \tru$, then we
know it was computed, and hence check if
$P(f)(\eqcl{z_1}{\equiv},\ldots \eqcl{z_k}{\equiv})$ is defined. If it
is \emph{not}, then we reject the word as it is not memoizing and
therefore, not coherent.

An $\passume(x=y)$ statement is dealt with as follows. Let
  us assume that the state of $\coh{\Aa}$ is $q = (\equiv, \emptyset, P, E)$
  when it processes $\passume(x=y)$. First $\coh{\Aa}$ checks if the
  equality is early as follows. We will say that a variable $v$ is an
  immediate superterm of $u$ in $q$, if there is a function $f$ and
  arguments $\vec{z}$ such that $P(f)(\vec{z}) = \eqcl{v}{\equiv}$ and
  $\eqcl{u}{\equiv} \in \vec{z}$. More generally, $v$ is a superterm
  of $u$ in $q$ if either $\eqcl{v}{\equiv} = \eqcl{u}{\equiv}$, 
  or there is a variable $w$ such that $w$
  is an immediate superterm of $u$ in $q$ and $v$ (inductively) is a
  superterm of $w$ in $q$. To check that $\passume(x=y)$ is \emph{not early},
  we will first check if there is a superterm $u$ of either $x$ or $y$ in $q$
  and a tuple $\vec{z}$ such that $\eqcl{u}{\equiv} \in \vec{z}$
such that $P(f)(\vec{z})$ is undefined but
  $E(f)(\vec{z}) = \top$. If $\passume(x=y)$ is early then $\coh{\Aa}$
  will merge several equivalence classes of variables as it performs
  congruence closure locally.
%
%
Whenever two equivalence classes $C$ and $C'$ merge, 
we set $E(f)(\vec{z})$ to $\top$ and appropriately define
$P(f)(\vec{z})$ (similar to the transition in $\feas{\Aa}$)
if there were any equivalent variables that were to $\top$ by $E(f)$.
%

After reading each letter, if the word is not rejected, we are
guaranteed that the word is coherent, and hence the meaning of the
state of automaton $\feas{\Aa}$ is correct when reading the next
letter. Hence the above automaton precisely accepts the set of
coherent words. 
\end{proof}

Observe that since $\coh{\Aa}$ constructed in \thmref{coherent-reg} is
deterministic, it can be modified without blowup to accept only
non-coherent words as well.

Using $\coh{\Aa}$ we can get a $\pspc$ algorithm to check if a program
is coherent.

We can now compute, given a program $s$, the NFA for $\exec(s)$, and
check whether the intersection of $\exec(s)$ and the above automaton
constructed for accepting non-coherent words is empty. Hence

\begin{theorem}
\thmlabel{dec-coherence}
Given a program $s$, one can determine if $s$ is coherent in $\pspc$.
\end{theorem}
\begin{proof}
Let $\Aa_s$ be the NFA accepting the set of execution of program
$s$. Let $\overline{\coh{\Aa}}$ be the automaton accepting the
collection of all non-coherent executions. Notice that $s$ is coherent
iff $L(\Aa_s) \cap L(\overline{\coh{\Aa}}) = \emptyset$. The $\pspc$
algorithm will construct the product of $\Aa_s$ and
$\overline{\coh{\Aa}}$ on the fly, while it searches for an accepting
computation.
\end{proof}


\section{Undecidability of Verification of Uninterpreted Programs}
\seclabel{undecidability}

We show that verifying uninterpreted programs is 
undecidable by reducing the halting problem for 2-counter machines. 



A 2-counter machine is a finite-state machine 
(with $Q$ as the set of states) augmented with
two counters $C_1$ and $C_2$ that take values in $\nats$.
At every step, the machine moves to a new state and performs
one of the following operations on one of the counters:
check for zero, increment by $1$, or decrement by $1$.
We can reduce the halting problem for 2-counter machines 
to verification of uninterpreted programs to show the following result
(detailed proof in Appendix~\ref{app:undec-ver-proofs}).

\begin{theorem}
\thmlabel{undec_main}
The verification problem for uninterpreted programs is undecidable.
\end{theorem}

The reduction in the above proof proceeds by encoding  
configurations using primarily three variables, a variable $x_\curr$ 
modeling the current state, and two variables $y_1$ and $y_2$ 
modeling the two counters, along with other variables $\set{x_q}_{q\in Q}$ modeling constants
for the set of states and the constant $0 \in \nats$ 
and few other auxiliary variables.

The key idea is to model a counter value $i$ using the term $f^i(0)$, and to ensure the data model has two functions $f$ and $g$, modeling increment and decrement functions respectively, that are inverses of each other on terms representing counter values.
We refer the reader to the details of the proof, but note here that this reduction in fact creates programs that are not memoizing nor have early assumes, and in fact cannot be made coherent with any bounded number of ghost variables.

In the following undecidability results (\thmref{undec_memo} and \thmref{undec_ea}) 
that argue that both our restrictions are required for decidability.
The proof ideas of these theorems are provided in Appendix~\ref{app:undec-ver-proofs}.

\begin{theorem}
\thmlabel{undec_memo}
The verification problem for uninterpreted programs all of whose executions are memoizing is undecidable.
\end{theorem}

\begin{theorem}
\thmlabel{undec_ea}
The verification problem for uninterpreted programs all of whose executions have early-assumes is undecidable.
\end{theorem}

The proofs for the theorems above also give reductions from the 
halting problem for 2-counter machines;
one ensures that executions satisfy the memoizing property
(while sacrificing the early-assumes criterion) 
and the other that executions satisfy early-assumes (but not memoizing).

Note that the above results do not, of course, mean that these two conditions 
themselves cannot be weakened while preserving decidability. 
They just argue that neither condition can be simply dropped.


%


\section{\texorpdfstring{$k$}{k}-Coherent Uninterpreted Programs}
\seclabel{kcoherence}

In this section, we will generalize the decidability results of
\secref{coherent-ver} to apply to a larger class of programs. 
Programs may sometimes not be coherent because they either
violate the memoizing property, i.e., they have executions 
where a term currently not stored in any of the program variables, is recomputed,
or, they violate the early assumes criterion, i.e., have executions
where an assume  is seen after some superterms have been dropped entirely.
However, some of these programs could be \emph{made} to coherent, if they
were given access to additional, auxiliary variables that store the
terms that need to be recomputed in the future or are needed until a 
relevant $\passume$ statement is seen in the future. For example, we
observed that program $P_2$ in \figref{example1} is not coherent
(\exref{coherence}), but program $P_3$, which is identical to $P_2$
except for the use of auxiliary variable $\cd{g}$ is coherent. These
auxiliary variables, which we call \emph{ghost variables}, can only be
written to. They are never read from, and so do not affect the
control flow in the program. We show that the verification
problem for \emph{$k$-coherent programs} --- programs that can be made
coherent by adding $k$ ghost variables --- is decidable. In addition,
we show that determining if a program is $k$-coherent is also
decidable.

Let us fix the set of program variables to be $V = \set{v_1,v_2,\ldots
  v_r}$ and the signature to be $\Sigma = (\Cc,\Ff)$ such that
$\init{V} \subseteq \Cc$ is the set of initial values of the program variables. 
Recall that executions of such programs are over the alphabet $\Pi =
\setpred{ \dblqt{x \passign y}, \dblqt{x \passign f(\vec{z})},
  \dblqt{\passume (x=y)}, \dblqt{\passume (x\neq y)}}{x, y, \vec{z}
  \textit{~are~in~} V}$. 

Let us define executions that use program variables $V$ and
\emph{ghost variables} $G = \set{g_1,g_2,\ldots g_l}$. These will be
words over the alphabet $\Pi(G) = \Pi \cup \setpred{\dblqt{g \passign
    x}}{g \in G \mbox{ and } x \in V}$. Notice, that the only
additional step allowed in such executions is one where a ghost
variable is assigned the value stored in a program variable. Before
presenting the semantics of such executions, we will introduce some
notation. For an execution $\rho \in \Pi(G)^*$, $\proj{\rho}{\Pi}$
denotes its \emph{projection} onto alphabet $\Pi$, i.e., it is the
sequence obtained by dropping all ghost variable assignment steps.

To define the semantics of such executions, we once again need to
associate terms with variables at each point in the execution. The
(partial) function $\comp : \Pi(G)^* \times (V \cup G) \to \Terms$
will be defined in the same manner as in \defref{computation}. The main
difference will be that $\comp$ now is a partial function, since we
will assume that ghost variables are undefined before their first
assignment; we do not have constants in our signature $\Sigma$
corresponding to initial values of ghost variables. For variables $x
\in V$, we define $\comp(\rho,x) = \comp(\proj{\rho}{\Pi},x)$, where
the function $\comp$ on the right hand side is given in
\defref{computation}. For ghost variables, $\comp$ is defined
inductively as follows.

\[
\begin{array}{rcll}
\comp(\epsilon,g) \!\! &=& \!\! \undf & \mbox{for each } g \in G\\
\comp(\rho\cdot \dblqt{g \passign x}, g) \!\! &=& \!\! \comp(\rho,x)\\
\comp(\rho\cdot \dblqt{g \passign x}, g') \!\! &=& \!\! \comp(\rho,g')
  & \mbox{for } g' \neq g\\
\comp(\rho\cdot a, g) \!\! &=& \!\! \comp(\rho,g) & \mbox{for any } g \in G, 
    \mbox{ and } a \in \Pi
\end{array}
\]

The set of equality and disequality assumes for executions with ghost
variables can be defined in the same way as it was defined for regular
executions. More precisely, for any execution $\rho \in \Pi(G)^*$, the
set of equality assumes is $\alpha(\rho) = \alpha(\proj{\rho}{\Pi})$
and the set of disequality assumes is $\beta(\rho) =
\beta(\proj{\rho}{\Pi})$.

An execution $\rho$ with ghost variables is \emph{coherent} if it
satisfies the memoization and early assumes conditions given in
\defref{coherence}, except that we also allow variables in $G$ to hold superterms modulo congruence. 
We are now ready to define the notion of
\emph{$k$-coherence} for executions and programs.
\begin{definition}
\deflabel{kcoherence} An execution $\sigma \in \Pi^*$ over variables
$V$ is said to be \emph{$k$-coherent} if there is an execution $\rho
\in \Pi(G)$ over $V$ and $k$ ghost variables $G = \set{g_1,\ldots
  g_k}$ such that
\begin{enumerate*}[label=(\alph*)]
\item $\rho$ is coherent, and
\item $\proj{\rho}{\Pi} = \sigma$.
\end{enumerate*}

A program $s$ over variables $V$ is said to be $k$-coherent if every
execution $\sigma \in \exec(s)$ is $k$-coherent.
\end{definition}

Like coherent executions, the collection of $k$-coherent executions
are regular.

\begin{proposition}
\proplabel{kcoherent-reg} 
The collection of $k$-coherent executions over program variables $V$
is regular.
\end{proposition}
\begin{proof}
\thmref{coherent-reg} establishes the regularity of the collection of
all coherent executions. 
The automaton $\coh{\Aa}$ constructed in its
proof in \secref{dec-coherence} essentially also recognizes the
collection of all coherent executions over the set $V \cup G$.
Now, observe that the collection of $k$-coherent
executions over the set $V$ is $\proj{L(\coh{\Aa})}{\Pi}$, and therefore the
proposition follows from the fact that regular languages are closed
under projections.
\end{proof}

Given a program, one can decide if it is $k$-coherent.
\begin{theorem}
\thmlabel{dec-kcoherence}
Given an uninterpreted program $s$ over variables $V$, one can
determine if $s$ is $k$-coherent in space that is linear in $|s|$ and
exponential in $|V|$.
\end{theorem}
\begin{proof}
Let $\Aa_s$ be the NFA that accepts the executions of program $s$, and
$\coh{\Aa}$ be the automaton recognizing the set of all coherent
executions over $V$ and $G$. Let $\Aa_{kcc}$ be the automaton
recognizing $\proj{L(\coh{\Aa})}{\Pi}$, which is the collection of all
$k$-coherent executions. Observe that $s$ is $k$-coherent if $L(\Aa_s)
\subseteq L(\Aa_{kcc})$. $\Aa_{kcc}$ is a nondeterministic automaton
with $O(2^{|V|^{O(1)}})$ states, and the proposition therefore
follows.
\end{proof}

Finally, the verification problem for uninterpreted $k$-coherent
programs is $\pspc$-complete.
\begin{theorem}
\thmlabel{verifying-kcoherent}
Given a $k$-coherent program $s$ with postcondition $\fals$, the
problem of verifying $s$ is $\pspc$-complete.
\end{theorem}
\begin{proof}
Since every coherent program is also $k$-coherent (for any $k$), the
lower bound follows from \thmref{verifying-coherent}. Let $\Aa_s$ be
the automaton accepting executions of $s$, $\coh{\Aa}$ the automaton
accepting coherent executions over $V$ and $G$, and $\feas{\Aa}$ the
automaton checking feasibility of executions. Observe that the
automaton $\Aa$ recognizing $\proj{(L(\coh{\Aa}) \cap
  L(\feas{\Aa}))}{\Pi}$ accepts the collection of $k$-coherent
executions that are feasible. The verification problem requires one to
determine that $L(\Aa) \cap L(\Aa_s) = \emptyset$. Using an argument
similar to the proof of \thmref{verifying-coherent}, this can be
accomplished in $\pspc$ because $\Aa$ has exponentially many states.
\end{proof}


\section{Verification of  Coherent and \texorpdfstring{$k$}{k}-Coherent Recursive Programs}
\seclabel{vpa}

In this section, we extend the decidability results to recursive
programs. In particular, we define the notions of coherence and
$k$-coherence for recursive programs, and show that the following
problems are decidable: (a) verifying recursive coherent programs; (b)
determining if a recursive program is coherent; (c) verifying
recursive $k$-coherent programs; and (d) determining if a program is
$k$-coherent.

This section will extend the automata-theoretic constructions to 
\emph{visibly pushdown automata}~\cite{Alur2004vpa}, 
and use the fact that they are closed under intersection; 
we assume the reader is familiar with these automata as well as with
\emph{visibly context-free languages}~\cite{Alur2009vpa}.

\subsection{Recursive Programs}

Let us fix a finite set of variables $V = \set{v_1, \ldots, v_r}$ 
and a finite set of method names/identifiers $M$.
Let $m_0 \in M$ be a designated ``main'' method.
Let us also fix a permutation  $\angular{v_1, \ldots, v_r}$
of variables and denote it by $\fixperm{V}$.

For technical simplicity and without loss of generality, we will
assume that for each method $m \in M$, the set of local variables is
exactly the set $V$; methods can, however, ignore certain variables if
they use fewer local variables.  We will also assume that each method
always gets called with \emph{all} the $r$ variables, thereby initializing all
local variables.  This also does not lead to any loss in 
generality---if a constant $c$ is used\footnote{Recall
  that, we model constants as initial values of certain variables.} 
  in some function call,
the caller can pass this constant to the called function, which
will use the passed parameter to initialize its local 
copy of the variable reserved for $c$.  
Finally, we will assume that when a method $m$ is
invoked, the order of the parameters is fixed to be $\fixperm{V}$.
This again does not lead to a loss of generality --- the caller can
rearrange the variables to the right order (by swapping) and then
reassign them after $m$ returns.  These conventions simplify the
exposition considerably.  We will, however, allow functions to return
multiple values back, and allow the caller to assign these to local
variables on return.

For every method $m$, let us fix a tuple $\out{m}$ 
of \emph{output variables} over $V$.
We require the output variables in $\out{m}$ to be distinct
(in order to avoid implicit aliasing that can be 
caused when  variables are repeated).






The syntax of recursive programs now has method definitions,
where the body of the methods can also include recursive calls,
besides the usual assignment, sequencing, conditionals and loops.
\begin{align*}
\pgm ::=& 
\,\, m \poutputs \out{m} \, \stmt \, 
\mid \, \pgm \, \pgm \\
\stmt ::=& 
\,\,  \pskip \, 
\mid \, x \passign y \, 
\mid \, x \passign f(\vec{z}) \, 
\mid \, \passume \, (\cond) \,
\mid \, \stmt \, ;\, \stmt \\
&
\mid \, \pif \, (\cond) \, \pthen \, \stmt \, \pelse \, \stmt \,
\mid \, \pwhile \, (\cond) \, \stmt 
\mid \, \vec{w} \passign m(\fixperm{V}) \\
\cond ::=& 
\, x = y \,
\mid \, x \not = y \,
\end{align*}

Here, $m$ is a method in $M$, and the variables
$x,y,\vec{z}, \vec{w}$ belong to $V$.  The length of the vector $\vec{w}$ must
of course match the length of the vector $\out{m}$ of output
parameters of the called method $m$.  A program consists of a
definition for each method $m \in M$, and we assume each method is
defined exactly once.


\begin{example}



\begin{figure}[t]
\begin{minipage}{0.5\textwidth}
$m_0 \poutputs \angular{\cd{b}}\, \code{\{}$ \\
\rule[1mm]{0.4cm}{0pt}
\passume \code{(T} $\neq$ \code{F);} \\
\rule[1mm]{0.4cm}{0pt}
\code{d} $\passign$ \code{key(x);} \\
\rule[1mm]{0.4cm}{0pt}
\pif \code{(d = k)} \pthen\, \code{\{}  \\ 
\rule[1mm]{0.8cm}{0pt}
\code{b $\passign$ T;} \\
\rule[1mm]{0.4cm}{0pt}
\code{\} } \\
\rule[1mm]{0.8cm}{0pt}
\pelse\, \code{\{ } \\
\rule[1mm]{0.9cm}{0pt}
\code{y $\passign$ x;}\\
\rule[1mm]{0.9cm}{0pt}
\code{x $\passign$ left(x);}\\
\rule[1mm]{0.9cm}{0pt}
\code{b $\passign$ $m_0(\cd{x}, \cd{k}, \cd{b}, \cd{d}, \cd{y}, \cd{T}, \cd{F})$;}\\
\rule[1mm]{0.9cm}{0pt}
\pif \code{(b = F)} \pthen \, \code{\{} \\
\rule[1mm]{1.2cm}{0pt}
\code{x $\passign$ right(x);}\\
\rule[1mm]{1.2cm}{0pt}
\code{b $\passign$ $m_0(\cd{x}, \cd{k}, \cd{b}, \cd{d}, \cd{y}, \cd{T}, \cd{F})$;}\\
\rule[1mm]{0.9cm}{0pt}
\code{\}}\\
\rule[1mm]{0.4cm}{0pt}
\code{\}} \\
\code{\}}
\end{minipage}
\caption{Example of uninterpreted recursive program.}
\figlabel{recexample}
\end{figure}

The example in~\figref{recexample} illustrates a recursive program with a single
method $m_0$.  This program checks whether any node reachable from
$\cd{x}$ using $\cd{left}$ and $\cd{right}$ pointers (which defines a
directed acyclic graph) contains a node with key $\cd{k}$.  
The method returns a single value---value of the variable $\cd{b}$
upon return.
The variables $\cd{x}$ and $\cd{k}$ are the true parameters, 
but we additionally augment the other variables
$\cd{b,d,y,T,F}$ for simplifying notations and have the method
rewrite those variables as described before.
Here $\cd{T}$ and $\cd{F}$ are variables storing the constants $\cd{true}$ 
and $\cd{false}$, respectively. 
Notice that, it is hard to find an iterative program with a 
bounded number of variables and without recursive functions 
that achieves the same functionality.
\end{example}

The semantics of recursive programs given
by the grammar $\pgm$ is a standard call-by-value semantics. 
We next define the formal semantics using terms over a data model.
\subsection{Executions}

Executions of recursive programs over the finite 
set of variables $V$ and the finite set of methods $M$ 
are sequences over the alphabet
$\Pi_M = \setpred{ 
\dblqt{x \passign y}, 
\dblqt{x \passign f(\vec{z})},  
\dblqt{\passume (x=y)}, 
\dblqt{\passume (x\neq y)}, 
\dblqt{\pcall \; m},  
\dblqt{\vec{w} \passign \preturn} 
}
{x, y, \vec{z}, \vec{w} \text{~are~in~} V, m \in M}$.

We will, in fact, treat the above alphabet as 
partitioned into three kinds: a \emph{call-alphabet}, 
a \emph{return-alphabet}, and an \emph{internal-alphabet}. 
The letters of the form $\dblqt{\pcall \; m}$ 
belong to the call-alphabet, 
the letters of the form
$\dblqt{\vec{z} \passign \preturn}$ belong to
the return alphabet, and the remaining letters belong to the internal alphabet.

The collection of all executions, denoted $\exec$, is given by the
following context-free grammar with start variable $E$.
\begin{align*}
E \goesto & \dblqt{x \passign y} 
\mid \dblqt{x \passign f(\vec{z})} 
\mid \dblqt{\passume (x=y)}
\mid \dblqt{\passume (x\neq y)} \\
 & \mid \dblqt{\pcall\;m} \,\cdot\, E \,\cdot\, \dblqt{\vec{w} \passign \preturn} 
\mid E \cdot E
\end{align*}

In the above rule, $m$ ranges over $M$. 
Furthermore, with respect to the call-return-internal alphabet defined above, 
the above defines a \emph{visibly pushdown language}.

\begin{definition}[Complete and Partial Executions of a recursive program]
\emph{Complete executions} of recursive programs 
that manipulate a set of variables $V$ are sequences over $\Pi_M$ 
and are defined as follows. 
Let $P$ be a recursive program. 
For each method $m \in M$, we denote by $s(m)$
the body (written over the grammar $\stmt$) in the definition of $m$.

Consider a grammar where we have nonterminals of the form $S_s$,
for various statements $s \in \stmt$,
where the rules of the grammar are as follows.

\begin{align*}
\begin{array}{rcl}
S_{\epsilon} &\goesto& \epsilon\\
S_{\pskip\,;\,s} & \goesto & S_s\\
S_{ x \passign y\,;\,s} &\goesto & \dblqt{x \passign y} \cdot S_s\\
S_{ x \passign f(\vec{z})\,;\,s} &\goesto& \dblqt{x \passign f(\vec{z})} \cdot  S_s \\
S_{ \passume (c)\,;\,s} &\goesto & \dblqt{\passume(c)}\cdot S_s\\
S_{\pif~(c)~\pthen~s_1~\pelse~s_2\,; ~~s} &\goesto& \dblqt{\passume(c)} \cdot  S_{s_1\,;\,s} \quad | \quad \dblqt{\passume(\neg c)} \cdot  S_{s_2\,;\,s} \\
S_{\pwhile~(c)\{s_1\}\,;~s} &\goesto& \dblqt{\passume(c)} \cdot  S_{s_1\,; ~\pwhile~(c)\{s_1\}\,;~s} \quad | \quad \dblqt{\passume( \neg c)} \cdot  S_{s}\\
S_{\vec{w} \passign m(\fixperm{V})\,;~s} &\goesto& \dblqt{\pcall~ m} \cdot S_{s(m)} \cdot \dblqt{\vec{z} \passign \preturn}
\cdot S_s
\end{array}
\end{align*}



\noindent
The set of executions of a program $P \in \pgm$, 
$\exec(P)$ are those accepted by the above grammar with start symbol  
$S_{s(m_0)\,;\,\epsilon}$, where
 $s(m_0)$ is the body of the ``main'' method $m_0$.
The set of \emph{partial executions}, denoted by $\pexec(P)$, is the
set of prefixes of complete executions in $\exec(P)$.
\end{definition}

In the above definition, the grammar for the language $\exec(P)$ 
is taken to be the one that can be defined by
using the minimal set of nonterminals for the definitions
$S_{s(m)\,;\,\epsilon}$, where $m \in M$.  
It is easy to see that this is a finite set
of nonterminals, and hence the above grammar is a context-free grammar. 
In fact, all productions rules except the one involving method calls
(i.e., production rules for non-terminals of the form
$S_{\vec{w} \passign m(\fixperm{V})\,;~s}$) are
right-regular grammar productions.
Further, the production rules for method calls
have a call-letter and return-letter guarding the first nonterminal.
Therefore, it is easy to see that the above defines a visibly pushdown language~\cite{Alur2009vpa}. 
A visibly pushdown automaton (VPA) that is at most quadratic
in the size of the program accepts this language as follows. This VPA
will have states of the form $S^m_s$ and mimic the right-regular
grammar productions using internal transitions generating the
associated terminal, and the rule for method calls by pushing the
nonterminal to execute after return onto the stack, and recovering it
in its state after the pop when simulating the return from the method
call.  This construction is fairly standard and simple, and we omit formal
definitions.


\subsection{Semantics of Recursive Programs and The Verification Problem}

\subsection*{Terms Computed by an Execution}

Let us now define the term computed for any (local) variable at any
point in the computation. We say a subword $\sigma$ of an execution is
\emph{matched} if $\sigma$ has an equal number of call-letters and
return-letters.



We now define the terms that correspond to local variables in scope
after a partial execution $\rho$.

\begin{definition}
We define $\comp : \exec \times V \to \Terms$ inductively as follows 
\begin{flalign*}
\begin{array}{rcll}
 \comp(\epsilon, x) \!\! &=&  \!\! \init{x} & 
 \\
 \comp(\rho \cdot \dblqt{x \passign y} , x) \!\! &=& \!\! \comp(\rho, y) &  
 \\
 \comp(\rho \cdot \dblqt{x \passign y} , x') \!\! &=& \!\! \comp(\rho, x') &x' \neq x
 \\
 \comp(\rho \cdot \dblqt{x \passign f(\vec{z})}, x) \!\! &=& \!\! f(\comp(\rho, z_1), \ldots, \comp(\rho, z_r))
 & 
 \vec{z} = (z_1, \ldots, z_r)
 \\
 \comp(\rho \cdot \dblqt{x \passign f(\vec{z})}, x') \!\! &=& \!\! \comp(\rho, x') & x' \neq x \\
 \comp(\rho \cdot \dblqt{\passume(R(\vec{z}))}, x) \!\! &=& \!\! \comp(\rho, x) & \\
 \comp(\rho \cdot \dblqt{\pcall \,m}, x) \!\! &=& \!\! \comp(\rho, x) &  \\
 \begin{aligned}\comp(\rho \cdot \dblqt{\pcall \,m} \cdot \rho' \quad\quad \quad \quad \quad\\\cdot \dblqt{\angular{w_1, \ldots w_r} \passign \preturn}, w_i)\end{aligned}\!\!& = &\!\!\!\begin{aligned}\comp(\rho \cdot \dblqt{\pcall \, m} \cdot \rho', \out{m}[i])\end{aligned} & 
 \rho' \text{ is matched} \\
\begin{aligned}\comp(\rho \cdot \dblqt{\pcall \,m} \cdot \rho' \quad\quad \quad \quad \quad\\\cdot \dblqt{\angular{w_1, \ldots w_r} \passign \preturn}, x)\end{aligned}\!\!& = &\!\!\begin{aligned}\comp(\rho, x)\end{aligned} 
&
\begin{aligned}
 x \not\in \set{w_1,\ldots w_r},\\
\rho' \text{ is matched}
\end{aligned}
\end{array}
\end{flalign*}

\medskip

The set of \emph{terms computed} by an execution $\rho$ is 
$\Terms(\rho) = \bigcup\limits_{\substack{\rho' \text{ is a prefix of } \rho,\\ v \in V}} \comp(\rho',v)$.

\end{definition}

\subsection*{Semantics of Recursive Programs}
Semantics of programs are, again, with respect to a data-model.  We
define the maps $\alpha$ and $\beta$ that collect the assumptions of
equality and disequality on terms, as we did for programs without
function calls; we skip the formal definition as it is the natural
extension to executions of recursive programs, ignoring the call and
return letters. An execution is feasible on a data model if these
assumptions on terms are satisfied in the model.

\subsection*{Verification Problem} 
As before, without loss of generality we
can assume that the post condition is $\fals$ (false). The
\emph{verification problem} is then: 
given a program $P$, determine if there is some
(complete) execution $\rho$ and data model $\Mm$ such that $\rho$ is
feasible on $\Mm$.

\subsection*{Coherence}
The notion of coherence is similar to the one for
non-recursive programs and their executions. In fact, it is precisely the definition of coherence for regular programs 
(\defref{coherence}), except for the fact that it uses the new definitions of $\comp$,
and $\alpha$ for recursive programs. 
We skip repeating the definition. 
Note that, in this case, the \emph{memoizing} condition and the \emph{early assumes} condition
are based on the set $\Terms(\sigma)$ (where $\sigma$ is a partial execution), 
which also includes all terms computed before, including
those by other methods. A recursive program is said to be
\emph{coherent} if all its executions are coherent.


We can now state our main theorems for recursive programs.

\begin{theorem}
The verification problem for coherent recursive programs is decidable, and is {\sc Exptime}-complete.\qed
\end{theorem}

The proof of the above result proceeds by constructing a VPA $\Aa_P$
that accepts the executions of the program $P$ and a VPA $\rfeas{\Aa}$
that accepts feasible executions of recursive programs, and checking
if $L(\Aa_P \cap L(\rfeas{\Aa}))$ is empty, the latter being a
decidable problem.

The automaton $\rfeas{\Aa}$ is designed similar to the automaton
$\feas{\Aa}$ constructed in
\secref{verifying-coherent-programs}. Below, we sketch the primary
ideas for handling the extension to recursive program executions.

\begin{proof}[Proof Sketch]
Recall that the automaton for nonrecursive program executions keeps
track of (a) an equivalence relation $\equiv$ over the variables $V$,
(b) partial maps for each $k$-ary function that map from
$(V/\equiv)^k$ to $V/\equiv$, and (c) a set of disequalities over the
equivalence classes of $\equiv$. The automaton $\rfeas{\Aa}$ will keep
a similar state, except that it would keep this over \emph{double} the
number of variables $V \cup V'$, where $V' = \{v' \mid v \in V\}$. The
variables in $V'$ correspond to terms in the \emph{caller} at the time
of the call, and these variables do not get reassigned till the
current method returns to its caller.

When the automaton sees a symbol of the form $\dblqt{\pcall\; m}$, 
it pushes the current state of the automaton on the stack, 
and moves to a state that has only the equivalence classes of the current 
$V$ variables (along with the partial functions and disequalities restricted to them). 
It also makes each
$v'$ equivalent to $v$. When processing assignments, assumes, etc. in
the called method, the variables $V'$ will never be reassigned, and
hence the terms corresponding to them will not be dropped.  At the end
of the method, when we return to the caller reading a symbol of the
form $\dblqt{\vec{w}\passign \preturn}$, we pop the state from the stack and
\emph{merge} it with the current state.

This merging essentially recovers the equivalence classes on variables
that were not changed across the call and sets up relationship
(equivalence, partial $f$-maps, etc.) to the variables $\vec{w}$
assigned by the return. This is done as follows.  Let the state popped
be $s'$ and the current state be $s$. Let us relabel variables in
$s'$, relabeling each $v \in V$ to $\underline{v}$ and each $v' \in V$
to $\underline{v'}$.  Let $\underline{V} = \{ \underline{v} \mid v \in
V \}$ and $\underline{V'} = \{ \underline{v'} \mid v' \in V' \}$.  Now
let us take the \emph{union} of the two states $s$ and $s'$
(inheriting equivalence classes, partial function maps, disequalities), to get $s
\oplus s'$ over variables $\underline{V'} \cup \underline{V} \cup V'
\cup V$.  In this structure, we merge (identify) each node
$\underline{v}$ with $v'$, retaining its label as $v'$.  Merging can
cause equivalence classes to merge, thereby also
updating partial function interpretations and disequalities.
We now \emph{drop} the variables $V'$, dropping the
equivalence classes if they become empty. The new state is over
$\underline{V'} \cup V$, and we relabel the variables $\underline{v'}$
to $v'$ to obtain the actual state we transition to. The $f$-maps and
set of disequalities get updated across these manipulations.

The resulting VPA has exponentially many states in $|V|$ 
and taking its intersection with the automaton $\Aa_P$ and 
checking emptiness clearly can be done in exponential time. 
The lower bound follows from the fact that checking reachability 
in recursive Boolean programs is already {\sc Exptime}-hard and 
the fact that we can emulate any
recursive Boolean program using a recursive uninterpreted
program (even with an empty signature of functions).
\end{proof}

We can also extend the notion of $k$-coherence to
executions of recursive programs; 
here we allow executions to have ghost assignments at any point
to local (write-only) ghost variables in scope, in order to make 
an execution coherent. We can build an automaton, again a VPA, that accepts \emph{all} $k$-coherent executions that are semantically feasible. 
Then, given a program $P$ and a $k \in \mathbb{N}$, we can build an automaton that accepts all coherent extensions of executions of $P$, and also check whether every execution of $P$ has at least one equivalent coherent execution. If this is true, then $P$ is $k$-coherent, and we can check whether the automaton accepts any word to verify $P$.

\begin{theorem}
The problem of checking, given a program $P$ and $k \in \mathbb{N}$, whether $P$ is $k$-coherent is decidable. And if $P$ is found to be $k$-coherent, verification of $P$ is decidable.
\end{theorem}


\section{Related Work}
\seclabel{related}

The class of programs (with and without recursion) over a finite set of Boolean variables 
admits a decidable verification 
problem~\cite{EsparzaKnoopFOSSACS99,AlurRSM2005,schwoon-phd02,Esparza2000,Godefroid2013}. 
As mentioned in the introduction, we believe that our work is the first natural 
class of programs that work over infinite data domains and also admits decidability
for the verification problem, without any severe restrictions on the structure of programs.

There are several automata-theoretic decidability results that could 
be interpreted as decidability results for programs---for example, coverability and reachability in (unsafe) Petri nets are decidable~\cite{Karp1969,Mayr1981,Kosaraju1982}, and this can be interpreted as a class of programs with counters with increments, decrements, and checks for positivity (but \emph{no} checks for zero), 
which is arguably not a very natural class of programs.
The work in~\cite{godoy09} establishes decidability of checking
equality assertions in uninterpreted \emph{Sloopy} programs with 
restricted control flow---such programs disallow the use of conditionals, 
loops and recursive calls inside other loops.

Complete automatic verification can be seen as doing 
both the task of finding inductive invariants and 
validating verification conditions corresponding to the various 
iteration/recursion-free snippets. 
In this light, there is classical work for certain domains 
like \emph{affine programs}, where certain static 
analyses techniques promise to always find an 
invariant, if there exists one that can be expressed in a 
particular logic~\cite{Karr1976,Muller-Olm2005,Granger1991,muller-olm2004}.
However, these results do not imply decidable verification for these programs, 
as there are programs in these classes that are correct but do not
have inductive invariants that fall in the fragment of logic considered. 

There is a line of work that takes an automata-theoretic flavor to 
verification~\cite{Matthias2013,Heizmann2010,Farzan2014,Farzan2015}, 
which rely on building automata accepting infeasible program traces, 
obtained by generalizing counterexample traces that can be proved 
infeasible through SMT solving. The method succeeds when
it can prove that the set of traces of the program that are erroneous,
are contained in the constructed set of infeasible traces.
The technique can handle several background theories, but of course tackles an undecidable problem. 
Our work relates to this line of work and can be interpreted as 
a technique for providing, directly and precisely, 
the set of infeasible coherent traces as a regular/visibly-pushdown language, 
and thereby providing decidable verification for programs with coherent traces. 
Combining our techniques for uninterpreted traces with the techniques above 
for other theories seems a promising future direction.

The theory of uninterpreted functions
is a fragment of first order logic with decidable quantifier free fragment~\cite{calcofcomputation}, and
 has been used popularly
in abstract domains in program 
analysis~\cite{Alpern1988,Gulwani2004,gulwani2004polynomial},
verification of hardware~\cite{Burch1994,Bryant2001} 
and software~\cite{Gulwani2006,Lopes2016}.

The notion of memoizing executions, which is an integral part of our coherence
and $k$-coherence definitions, is closely related
to \emph{bounded path-width}~\cite{robertson1983graph}. 
We can think of a computation of the program as 
sweeping the initial model using a window of terms 
defined by the set of program variables. 
The memoizing condition essentially says that the set of windows 
that contain a term must be a \emph{contiguous set}; 
i.e., a term computed should not be ``dropped'' if it will recomputed.
The notion of bounded path-width and the related notion 
of bounded tree-width have been exploited recently 
in many papers to provide decidability results in verification~\cite{chatterjee2015,chatterjee2016}.


\section{Conclusions}
\seclabel{conclusions}
We have proved that the class of coherent programs and $k$-coherent programs (for any $k \in \mathbb{N}$) admit decidable verification problems. Checking whether programs are coherent or $k$-coherent for a given $k$ is also decidable. Moreover, the decision procedure is not very expensive, and in fact matches the complexity of verifying the weaker class of Boolean programs over the same number of variables. 


Our results lay foundational theorems for decidable reasoning of uninterpreted programs, and open up a research direction exploring problems that can be tackled using uninterpreted functions/relations. There are several avenues for applications that we foresee.
One is reasoning about programs using uninterpreted abstractions, as in the work on reasoning with  containers~\cite{Dillig2011} and modeling pointers in heap manipulating programs\cite{natproofs2013,natproofs2014,natproofs2017}.
Such applications will likely call for an extension of our results to handle axioms that restrict the uninterpreted functions (such as associativity and commutativity of certain functions) or to incorporate first order theories such as arithmetic and sets. Specifications for heap manipulating programs often involve recursive definitions, and this may require enriching our results to incorporate such definitions. 
We also conjecture that our results can be useful in  domains such as verification of compiler transformations (such as instruction reordering), when proofs of correctness of transformations rely only on a few assumptions on the semantics of operations and library functions. Trace abstraction based verification approaches~\cite{Heizmann2010,Heizmann2009SAS,Farzan2014,Farzan2015} build automata that capture infeasible traces incompletely using a counter-example guided approach. In this context, our results would enrich such automata---we can accept precisely the set of infeasible traces that become infeasible when making functions uninterpreted. This is a possible direction to combine with other background theories in verification applications.

\appendix

\newpage
\section{Proofs from~\secref{verifying-coherent-programs}}
\label{app:proofs-automaton}

\subsection{Proof of~\lemref{coherent-superterm}}

\begin{proof}
This follows trivially from the observation 
that every prefix of a coherent execution
is coherent and thus has early assumes.
\end{proof}

\subsection{Proof of~\lemref{window_update_equation}}

\begin{proof}
First let us argue that $\sim_{s, s'}$ is consistent with the term model.
Consider $(t_1, t_2) \in W \times W$.
First it is easy to see that since $(\equiv_W, D_W, P)$ is consistent,
we have that if $t_1 \sim_{s, s'} t_2$, 
then $\termsem{t_1}{\termmod{E'}} = \termsem{t_2}{\termmod{E'}}$,
where $E' = E \cup \set{(s, s')}$.
We now prove the reverse direction.
Assume that $\termsem{t_1}{\termmod{E'}} = \termsem{t_2}{\termmod{E'}}$.
We will show that $t_1 \sim_{s, s'} t_2$ by inducting on
the step in the congruence closure algorithm
that adds $(t_1, t_2)$ in the congruence $\congcl{E'}$.
In the base case, (w.l.o.g) $t_1 \equiv_W s$ and $t_2 \equiv_W s'$,
in which case clearly we have  $t_1 \sim_{s, s'} t_2$.
In the inductive case, there is a $k$-ary function $f$
and terms $u^1_1, u^2_1, \ldots, u^k_1$ and $u^1_2, u^2_2, \ldots u^k_2$
with $t_1 = f(u^1_1, \ldots, u^k_1)$ and
$t_2 = f(u^1_2, \ldots, u^k_2)$,
and we have that $\termsem{u^i_1}{\termmod{E'}} = \termsem{u^i_2}{\termmod{E'}}$ for each $1 \leq i \leq k$.
Now, since $W$ is super-term closed with respect to $W$, $E$ and $(s, s')$,
we have that there are terms $v^1_1, v^2_1, \ldots, v^k_1, v^1_2, \ldots, v^k_2 \in W$
such that $u^i_j \congcl{E} v^i_j$ for all $i, j$.
Now, by inductive hypothesis, $v^i_1 \sim_{s, s'} v^i_2$ for all $1 \leq i \leq k$.
Further, consistency ensures that 
$P(f)(\eqcl{v^1_i}{\equiv_W}, \ldots, \eqcl{v^1_i}{\equiv_W}) = \eqcl{t_i}{\equiv_W}$
for $i \in \set{1, 2}$.
This means that our procedure for constructing $\sim_{s, s'}$ will imply that $t_1 \sim_{s, s'} t_2$.
%

The proof of consistency of $D'_W$ and $P'$ is rather straightforward
and follows from the fact that $\sim_{s, s'}$ is an equivalence
and that $(\equiv_W, D_W, P)$ is consistent.
\end{proof}

\subsection{Proof of~\lemref{window_check_eq}}

\begin{proof}
We will assume that $(t_1, t_2) \not\in \congcl{E}$ as otherwise the proof is straigtforward.
Let $E' = E \cup \set{(t_1, t_2)}$.
We first observe that for every $(s, s') \in \congcl{E'} \setminus \congcl{E}$,
we have that there are terms $r, r'$ such that $s \congcl{E} r$,
$s' \congcl{E} r'$, $r$ is a superterm of $t_1$ and $r'$ is a superterm of $t_2$.

Now, let us consider the case when $(s, s') \in \congcl{E'} \cap D$.
It follows from the previous paragraph and because $W$ is superterm closed
that there are terms $u, u' \in W$
such that $u \congcl{E} s$ and $u' \congcl{E} s'$.
Also, since $\sim_{t_1, t_2}$ is consistent with respect to $(W, E\cup \set{t_1, t_2}, D)$
(see \lemref{window_update_equation}), we must have that $u \sim_{t_1, t'_2} u'$.
Also since $D_W$ is consistent, we have $([u]_{\congcl{W}}, [u']_{\congcl{W}}) \in D_W$.
This proves the first direction.

The other direction is more straightforward and follows from consistency of $(\equiv_W, D_W, P)$ with
respect to $(W, E, D)$ and that $\sim_{t_1, t'_2}$ is consistent with respect to $(W, E\cup \set{t_1, t_2}, D)$
(from \lemref{window_update_equation}).
\end{proof}

\subsection{Proof of~\lemref{window_check_diseq}}

\begin{proof}
Direct consequence of consistency.
\end{proof}

\subsection{Proof of~\lemref{inv_automaton}}




\begin{proof}
We will induct on the length of the execution $\rho$.

\textbf{Base case}.
In the base case $\rho = \epsilon$ and $q_\rho = q_0$.
This execution is clearly feasible and $q_\rho \neq \reject$.
Also, $q_\rho$ is trivially consistent with respect to $(\setpred{\comp(\rho,x)}{x \in V}, \alpha(\rho), \beta(\rho))$.

\textbf{Inductive case.}
Here, $\rho = \sigma \cdot a$. 

Let $T_\gamma =\Terms(\gamma)$ be the set of terms computed by $\gamma$, where $\gamma \in \set{\rho, \sigma}$.
Let $E_\gamma = \alpha(\gamma)$, $D_\gamma = \beta(\gamma)$ and $W_\gamma = \setpred{\comp(\gamma, v)}{v \in V} \subseteq T_\gamma$ for  $\gamma \in \set{\rho, \sigma}$.
We have that $T_\sigma \subseteq T_\rho$, $E_\sigma \subseteq E_\rho$ and $D_\sigma \subseteq D_\rho$.

If $q_\sigma = \reject$, then $q_\rho = \reject$.
In this case, from the inductive hypothesis, we have that $\sigma$ is infeasible and 
thus $\congcl{E_\sigma} \cap D_\sigma \neq \emptyset$.
Since $\alpha(\sigma) \subseteq \alpha(\rho)$ and $\beta(\sigma) \subseteq \beta(\rho)$,
we have that $\congcl{E_\rho} \cap D_\rho \neq \emptyset$ and thus $\rho$ is also infeasible.

Otherwise, let us assume that $\sigma$ is feasible.
Let $q_\sigma = (\equiv, d, P)$.
and let $q_\rho = (\equiv', d', P')$.
Inductively, $q_\sigma$ is consistent with respect to $(W_\sigma, E_\sigma, D_\sigma)$.
Below, we prove that $q_\rho$ is consistent with respect to $(W_\rho, E_\rho, D_\rho)$.
If $d'$ is reflexive, then we will have $q_\rho = \reject$ instead.





Depending upon what the letter $a$ is, we have the following cases.
\begin{enumerate}
	\item \textbf{Case $a$ is $\dblqt{x \passign y}$}.\\
	In this case, $T_\rho = T_\sigma$, $E_\rho = E_\sigma$ and $D_\rho = D_\sigma$ and thus $\rho$ is also feasible.
	We will assume that $y \neq x$ 
	(otherwise $q_\rho = q_\sigma$, which is consistent with respect to $(W_\rho, E_\rho, D_\rho)$).

	Further, $\comp(\rho, v) = \comp(\sigma, v)$ for all $v \neq x$ and $\comp(\rho, x) = \comp(\sigma, y)$.

	Since $(\equiv, d, P)$ is consistent with respect to $(W_\sigma, D_\sigma, E_\sigma)$,
	we have that for every $x_1, x_2 \in V$ with $(x_1, x_2) \in \equiv$, 
	$\termsem{\comp(\sigma, v_1)}{\termmod{E_\sigma}} = \termsem{\comp(\sigma, v_2)}{\termmod{E_\sigma}}$.
	Consider a pair $(v_1, v_2) \in \equiv'$.
	Now, if $v_1 \neq x$ and $v_2 \neq x$, we have that $(v_1, v_2) \in \equiv$, and thus
	$\termsem{\comp(\rho, v_1)}{\termmod{E_\rho}} = \termsem{\comp(\rho, v_2)}{\termmod{E_\rho}}$.
	Otherwise, assume without loss of generality that $v_1 = x$ and $v_2 \neq x$.
	In this case we have $(v_2, y) \in \equiv$, and by consistency we have $\termsem{\comp(\rho, v_2)}{\termmod{E_\rho}} = \termsem{\comp(\sigma, v_2)}{\termmod{E_\sigma}} = \termsem{\comp(\sigma, y)}{\termmod{E_\sigma}} = \termsem{\comp(\rho, x)}{\termmod{E_\rho}}$.

	The consistency of $d'$ and $P'$ follows from the consistency of $q_\sigma$.
	Further, $d'$ is irreflexive because $d$ is irreflexive.

	\item
	\textbf{Case $a$ is `$x := f(z_1, \ldots z_k)$' where $k=arity(f)$}.\\
	Here, we have two cases:

		\begin{enumerate}
			\item \textbf{Case $P(f)(\eqcl{z_1}{\equiv}, \ldots, \eqcl{z_k}{\equiv}) = \eqcl{v}{\equiv}$ is defined.} \\
			Here, if $x \in \eqcl{v}{\equiv}$, then we have $\equiv' = \equiv$, $d' = d$ and $P' = P$.
			The resulting state is consistent with $(W_\rho, E_\rho, D_\rho)$.
			This is because $q_\sigma$ is consistent with $(W_\sigma, E_\sigma, D_\sigma)$ and thus
			$\termsem{\comp(\sigma, x)}{\termmod{E_\sigma}} = \termsem{\comp(\sigma, v)}{\termmod{E_\sigma}}$
			and thus $\termsem{\comp(\rho, x)}{\termmod{E_\rho}} = \termsem{\comp(\rho, v)}{\termmod{E_\rho}}$.
			The consistency of $d'$ and $P'$ follows similarly.

			Otherwise, we have that $x \not\in \eqcl{v}{\equiv}$. 
			The consistency of $q_\rho$
			again follows from the facts that 
			(a) $q_\sigma$ is consistent, 
			(b) $E_\sigma = E_\rho$, 
			(c) $D_\sigma = D_\rho$, and,
			(d) $\termsem{\comp(\rho, x)}{\termmod{E_\rho}} = \termsem{\comp(\rho, v)}{\termmod{E_\rho}}$.




			\item \textbf{Case $P(f)(\eqcl{z_1}{\equiv}, \ldots, \eqcl{z_k}{\equiv})$ is undefined} \\
			In this case, since $\rho$ is memoizing, it must be the case that there is no term $t \in T_\sigma$
			such that $\termsem{t}{\congcl{E_\sigma}} = \termsem{f(t_1, t_2, \ldots, t_k)}{\congcl{E_\sigma}}$,
			where $t_i = \comp(\sigma, z_i)$.
			This, in conjuction with the facts that (a) $E_\sigma = E_\rho$, (b) $D_\sigma = D_\rho$,
			and (c) $q_\sigma$ is consistent with respect to $(W_\sigma, E_\sigma,  D_\sigma)$ ensures the consistency of $q_\rho$.


		\end{enumerate}
		In both the above cases, feasibility of $\sigma$ implies feasibility of $\rho$ and this is also correctly reflected in $d'$, which is not reflexive since $d$ is irreflexive.

	\item 
	\textbf{Case $a$ is $\dblqt{\passume \, (x = y)}$}.\\
	In this case, consistency follows from \lemref{window_update_equation}.
	Further, if $\rho$ is infeasible, then $d'$ is not reflexive (\lemref{window_check_eq}), in which case we set $q_\rho = \reject$.

	\item 
	\textbf{Case $a$ is ``$\passume \, (x \neq y)$''}.\\
	In this case $T_\sigma = T_\rho, W_\sigma = W_\rho$, $E_\sigma = E_\rho$.
	Further, $D_\rho = D_\sigma \cup \set{(\comp(\sigma, x), \comp(\sigma, y))}$.
	Also, $\equiv' = \equiv$ and $P' = P$ and we just have argue for the consistency of $d'$ which is immediate
	from the fact that the new pair $(c_1, c_2) \in d' \setminus d$ (if it exists) is such that
	$x \in c_1$ and $y \in c_2$ for which we have the witness $(\comp(\rho, x), \comp(\rho, y)) \in D_\rho$.

	If $\rho$ is infeasible, then $d'$ becomes irreflexive (\lemref{window_check_diseq})


\end{enumerate}
\end{proof}

\newpage
\section{Appendix: Undecidability of Verification Problem}
\label{app:undec-ver-proofs}

We show the proof of \thmref{undec_main} that the verification of (unrestricted) uninterpreted programs, even without recursion, is undecidable. We do this through a reduction from the halting problem
for 2-counter machines to the verification of uninterpreted programs. In fact, our proof shows that verification of uninterpreted programs with even two unary functions is undecidable.

Let $M = (Q, s, \halt, C_1, C_2, \delta)$
be a 2-counter Minsky machine such that $Q$ is a finite set of states,
$s \in Q$ is the start, $\halt \in Q$ is a special state,
$C_1$ and $C_2$ are counters, both of which take values in $\nats$
and $\delta : Q \to \{\inc, \dec, \chkzero\} \times \{C_1, C_2\} \times Q$
is the deterministic transition function.

Let $\Sigma = (\Ff, \Rr, \Cc)$ be a signature with $\Rr = \Cc = \emptyset$
and $\Ff = \set{f, g}$ where $f$ and $g$ are both unary functions.
We will define a program $P_M$ over $\Sigma$ and an $\Ll_=$ formula $\phi_M$
such that $M$ reaches $\halt$ starting from $s$ iff
$P_M \models \phi_M$. \\

\noindent
\textbf{Variables.}
We first define the set of variables $V_M$ of the program.
For every state $q \in Q$, we have a variable $x_q$, 
which will behave like constants---they will never occur 
on the left hand side of any assignment statement of our program.
We will also have a variable $x_\curr$ that will keep track of
the current state of the machine in a run of the machine.
We will have two variables $y_1$ and $y_2$ that will correspond to
the values of the values of counters $C_1$ and $C_2$ respectively in a run.
We will have a variable $z$ which will correspond to the constant $0$.
Finally, we will also have two auxiliary (helper) variables 
$w_1, w_2, w_3$ and $w_4$.
Thus, we have $V_M = \set{x_q \,|\, q \in Q} \cup \set{x_\curr, y_1, y_2, z, w_1, w_2, w_3, w_4}$.\\

\noindent
\textbf{Statements.}
The first few statements in $P_M$ correspond to initialization:
$x_\curr$ is initialized to $x_{q_0}$ and both $y_1$ and $y_2$
get initialized to $z$.
Further we demand that all constants are different from each other:
for every $q \neq q'$ we
demand that $x_q$ and $x_{q'}$ are unequal,
and that $z$ is different from any $x_q$.
\begin{align*}
\code{INIT:} & & \\
&\passume \, (x_{q} \neq x_{q'}) & \quad \quad \texttt{(** for every } q\neq q' \texttt{**)} \\
&\passume \, (x_q\neq z) & \quad \quad \texttt{(** for every } q \texttt{**)} \\
&\passume \, (x_\curr := x_s) \\
&\passume \, (y_1 := z) \\
&\passume \, (y_2 := z) \\
\end{align*}

The initialization block is followed by
a while loop with the condition $(x_\curr \neq x_\halt)$.
The body of the loop contains, one block (labelled $\code{BLOCK}_q$) for every
for every $q \to (instr, C_i, q') \in \delta$.
The statements inside $\code{BLOCK}_q$ depend upon what $instr$ is:

\begin{enumerate}
\item \textbf{Case $instr$ is $\inc$ or $\dec$}:\\
In this case the code block is
	\begin{align*}
	\code{BLOCK}_q: &   \\
	&\pif \, (x_\curr =x_q) \{ & \\
	& \quad w_1  := f(y_i) \\
	& \quad w_2 := g(y_i) \\
	& \quad w_3 := g(w_1) \\
	& \quad w_4 := f(w_2) \\
	& \quad \passume \, (w_3 = y_i) \\
	& \quad \passume \, (w_4 = y_i) \\
	& \quad \stmt \\
	& \quad x_\curr := x_{q'}\\
	&  \} 
	\end{align*}
where $\stmt$ is $y_i := f(y_i)$ if $instr$ is $\inc$
and $y_i := g(y_i)$ if $instr$ is $\dec$.

\item \textbf{Case $instr$ is $\chkzero$}:\\
In this case the code block is
	\begin{align*}
	\code{BLOCK}_q: &   \\
	& \pif \, (x_\curr =x_q) \{ & \\
	& \quad \passume \, (y_i = z) \\
	&  \} 
	\end{align*}
\end{enumerate}


Finally the specification is $\phi_M \delequal (x_\curr \neq x_\halt) $

The correctness of the reduction is captured by the following lemma
\begin{lemma}
$M$ does not reach $\halt$ iff $P_M \models \phi_M$.
\end{lemma}

\begin{proof}
Before we begin the actual proof, we note that every complete execution of $P_M$
ends in $\dblqt{\passume(x_\curr = x_\halt)}$.

Let us first prove the forward direction: 
If $M$ does not reach $\halt$, then $P_M \models \phi_M$.
Let's say $M$ does not reach $\halt$. In this case,
we show that no complete execution of $P_M$ is feasible (and thus $P_M \models \phi_M$).
Assume, on the contrary, that there is a complete execution $\sigma$
which is also feasible.
As we noted before, $\sigma$ ends in $\dblqt{\passume(x_\curr = x_\halt)}$.
We first observe that $\sigma$ is feasible in $\Mm_\nats$ 
where the universe is the set of natural numbers,
$z$ is interpreted as $0$ and $f$ and $g$ are respectively $\lambda i. i+1$ and $\lambda i. i-1$.
Now, let $\code{INIT}, \code{BLOCK}_{q_{i_1}} \ldots \code{BLOCK}_{q_{i_k}}$
be the sequence of blocks taken in $\sigma$.
Since $\sigma$ is feasible, it must be that in the last block $\code{BLOCK}_{q_{i_k}}$,
$x_\curr$ must be assigned to $x_\halt$.
Now, consider the corresponding run $\rho = s \to q_{i_1} \cdots q_{i_k} \to \halt$.
It easily follows that $\rho$ is a valid run of $M$, contradicting our initial assumption that
$M$ does not reach $\halt$.

Let us now prove the other direction (in a contrapositive manner):
If $M$ reaches $\halt$, then $P_M \not\models \phi_M$.
Consider the run $\rho$ of $M$ that reaches $\halt$.
Then, $\rho = q_{i_0} \to q_{i_1} \cdots q_{i_k} \to \halt$, where $q_{i_0} = s$.
We now construct a \emph{feasible} complete execution $\sigma$ of $P_M$, which,
as we noted before, ends in $\dblqt{\passume(x_\curr = x_\halt)}$
(and thus $P_M \not\models \phi_M$).
The execution $\sigma$ will essentially mimic the run $\rho$ of $M$, and is the
sequence of statements generated by the sequence of blocks
$\code{INIT}, \code{BLOCK}_{q_{i_1}} \ldots \code{BLOCK}_{q_{i_k}}$.
Since $\rho$ is a run of $M$, it follows that $\sigma$ is feasible in the FO model
$\Mm_\nats$, where the universe is the set of natural numbers,
the constant $z$ is interpreted as $0$, $f$ is interpreted to be the
function $\lambda i. i+1$
 and $g$ is the function $\lambda i. i-1$.
\end{proof}

\subsection{Undecidability when programs executions have early-assumes but are not memoizing}
We now give the main modifications of the above proof to adapt it to a proof of \thmref{undec_memo}.

We will continue to have an input variable that stands for the counter value $0$, and $f$ and $g$ defined over it with $f^i(0)$ denoting the counter value $i$.
Again, we would like to assert that for each of the terms $f^i(0)$, $g(f(f^{i}(0)))=f^i(0)$. The assumes that we executed in the previous proof are not, however, early assumes.

In order to ensure the above property of the model, we introduce a ``fake'' counter variable $c^*$. The program always ensures, in every round of simulating the 2-CM one step, that $c^*$ always advances (i.e., we update $c^*$ to $f(c^*)$), and we assure that $c^* = g(f(c^*))$
using an assume. We remove the assumes that we did using the true counters $c_1$ and $c_2$
to ensure $g$ is the inverse of $f$.

Note that the counter $c^*$ always stays ahead of the other counters, and hence the assumes it makes are early. 

The second complication to deal with are the assumptions that $c_1$ and $c_2$ make when they check for equality with $0$.
These again are not early assumes in the construction. We modify the manipulation of the counter $c^*$ so that at every point it computes $r_1(c^*)$, $r_2(c^*)$, 
$s_1(c^*)$, and $s_2(c^*)$.
It also assumes the first two
terms are equal and the second two terms are disequal. 
Furthermore, for $0$ we have the reverse; $r_1(0) \not = r_2(0)$
and $s_1(0)=s_2(0)$.
These assumes are again early as $c^*$ stays ahead of the true computation. Now, when we want to simulate the check that $c_1$
is $0$, we instead check whether 
$r_1(c_1) \not = r_2(0)$, and
to check if $c_1$ is not $0$,
we check $s_1(c_1) \not = s_2(c_1)$. Hence all assumes on the real simulation are disequality assumes, and hence do not break the early assume requirement.  

The rest of the construction and proof is simple.

\subsection{Undecidability when programs executions are memoizing but do not have early-assumes}
We now give the main modifications of the main proof above to adapt it to a proof of \thmref{undec_ea}. In fact we will show programs with a \emph{single} non-early assume already are undecidable to verify.

We again model the counters on terms as above. However, we start with two variables $x$ and $y$, both representing $0$. 

In every step, we simulate the counters $c_1$ and $c_2$ on terms constructed from $x$, as before, but we make no assumes of equality. Consequently, the initial model on terms involving $x$ would be a ``free'' structure. Meanwhile, on $y$, we simulate a fake counter as before, that continually increases, staying always ahead of the computation on $x$, and ensuring that $g$ is the inverse of $f$ on terms involving $y$. 

The initial model will keep the two submodels from $x$ and $y$ disjoint. 

At the end of the computation (when we reach a final state of the 2-CM machine), we do a single assume $x=y$. This will collapse 
two models, and impose the structure the model on $y$ had on the free model on $x$. Hence the counters and simulation become a correct simulation.

We again need to handle assumes that we need to make to check counters equality with $0$. 
We do this as follows. 
We have functions $h_1$ and $h_2$ that operate on the free model involving $x$, and in particular we consider the terms of the form $h_1^n(x)$ and $h_2^n(x)$. These essentially are meant to keep a copy of $0$, i.e., $x$. The fact that these are equal to $x$ will be enforced later when the models of $x$ and $y$ merge.  Every time we want to check if a counter $c_i$ is equal to $x$, we check instead whether 
$r(h_i(0)) = r(c_i)$, where $r$ is yet another function we introduce. The idea is that $r$ will be identity on these terms, but again that is assured later.

Meanwhile, in $y$, we ensure that $h_1(y)=y$
and $h_1(y)=y$ and $r$ is identity. 
When we finally collapse the models with an assume
that $x$ and $y$ are equal, it will impose the required structure on the model of $x$, and hence the $0$-equality tests will become correct. 

The above is the main idea of the construction; the rest of the construction and proof is straightforward.

\begin{acks}                            
We gratefully acknowledge National Science Foundation for supporting
Umang Mathur (grant NSF CSR 1422798), P. Madhusudan (grants 1138994 and SHF 1527395) 
and Mahesh Viswanathan (grant NSF CPS 1329991).
\end{acks}

\bibliography{references}


\begin{thebibliography}{38}


\ifx \showCODEN    \undefined \def \showCODEN     #1{\unskip}     \fi
\ifx \showDOI      \undefined \def \showDOI       #1{#1}\fi
\ifx \showISBNx    \undefined \def \showISBNx     #1{\unskip}     \fi
\ifx \showISBNxiii \undefined \def \showISBNxiii  #1{\unskip}     \fi
\ifx \showISSN     \undefined \def \showISSN      #1{\unskip}     \fi
\ifx \showLCCN     \undefined \def \showLCCN      #1{\unskip}     \fi
\ifx \shownote     \undefined \def \shownote      #1{#1}          \fi
\ifx \showarticletitle \undefined \def \showarticletitle #1{#1}   \fi
\ifx \showURL      \undefined \def \showURL       {\relax}        \fi
\providecommand\bibfield[2]{#2}
\providecommand\bibinfo[2]{#2}
\providecommand\natexlab[1]{#1}
\providecommand\showeprint[2][]{arXiv:#2}

\bibitem[\protect\citeauthoryear{Alpern, Wegman, and Zadeck}{Alpern
  et~al\mbox{.}}{1988}]%
        {Alpern1988}
\bibfield{author}{\bibinfo{person}{B. Alpern}, \bibinfo{person}{M.~N. Wegman},
  {and} \bibinfo{person}{F.~K. Zadeck}.} \bibinfo{year}{1988}\natexlab{}.
\newblock \showarticletitle{Detecting Equality of Variables in Programs}. In
  \bibinfo{booktitle}{\emph{Proceedings of the 15th ACM SIGPLAN-SIGACT
  Symposium on Principles of Programming Languages}}
  \emph{(\bibinfo{series}{POPL '88})}. \bibinfo{publisher}{ACM},
  \bibinfo{address}{New York, NY, USA}, \bibinfo{pages}{1--11}.
\newblock
\showISBNx{0-89791-252-7}
\urldef\tempurl%
\url{https://doi.org/10.1145/73560.73561}
\showDOI{\tempurl}


\bibitem[\protect\citeauthoryear{Alur, Benedikt, Etessami, Godefroid, Reps, and
  Yannakakis}{Alur et~al\mbox{.}}{2005}]%
        {AlurRSM2005}
\bibfield{author}{\bibinfo{person}{Rajeev Alur}, \bibinfo{person}{Michael
  Benedikt}, \bibinfo{person}{Kousha Etessami}, \bibinfo{person}{Patrice
  Godefroid}, \bibinfo{person}{Thomas Reps}, {and} \bibinfo{person}{Mihalis
  Yannakakis}.} \bibinfo{year}{2005}\natexlab{}.
\newblock \showarticletitle{Analysis of Recursive State Machines}.
\newblock \bibinfo{journal}{\emph{ACM Trans. Program. Lang. Syst.}}
  \bibinfo{volume}{27}, \bibinfo{number}{4} (\bibinfo{date}{July}
  \bibinfo{year}{2005}), \bibinfo{pages}{786--818}.
\newblock
\showISSN{0164-0925}
\urldef\tempurl%
\url{https://doi.org/10.1145/1075382.1075387}
\showDOI{\tempurl}


\bibitem[\protect\citeauthoryear{Alur and Madhusudan}{Alur and
  Madhusudan}{2004}]%
        {Alur2004vpa}
\bibfield{author}{\bibinfo{person}{Rajeev Alur} {and} \bibinfo{person}{P.
  Madhusudan}.} \bibinfo{year}{2004}\natexlab{}.
\newblock \showarticletitle{Visibly Pushdown Languages}. In
  \bibinfo{booktitle}{\emph{Proceedings of the Thirty-sixth Annual ACM
  Symposium on Theory of Computing}} \emph{(\bibinfo{series}{STOC '04})}.
  \bibinfo{publisher}{ACM}, \bibinfo{address}{New York, NY, USA},
  \bibinfo{pages}{202--211}.
\newblock
\showISBNx{1-58113-852-0}
\urldef\tempurl%
\url{https://doi.org/10.1145/1007352.1007390}
\showDOI{\tempurl}


\bibitem[\protect\citeauthoryear{Alur and Madhusudan}{Alur and
  Madhusudan}{2009}]%
        {Alur2009vpa}
\bibfield{author}{\bibinfo{person}{Rajeev Alur} {and} \bibinfo{person}{P.
  Madhusudan}.} \bibinfo{year}{2009}\natexlab{}.
\newblock \showarticletitle{Adding Nesting Structure to Words}.
\newblock \bibinfo{journal}{\emph{J. ACM}} \bibinfo{volume}{56},
  \bibinfo{number}{3}, Article \bibinfo{articleno}{16} (\bibinfo{date}{May}
  \bibinfo{year}{2009}), \bibinfo{numpages}{43}~pages.
\newblock
\showISSN{0004-5411}
\urldef\tempurl%
\url{https://doi.org/10.1145/1516512.1516518}
\showDOI{\tempurl}


\bibitem[\protect\citeauthoryear{Bradley and Manna}{Bradley and Manna}{2007}]%
        {calcofcomputation}
\bibfield{author}{\bibinfo{person}{Aaron~R. Bradley} {and}
  \bibinfo{person}{Zohar Manna}.} \bibinfo{year}{2007}\natexlab{}.
\newblock \bibinfo{booktitle}{\emph{The Calculus of Computation: Decision
  Procedures with Applications to Verification}}.
\newblock \bibinfo{publisher}{Springer-Verlag}, \bibinfo{address}{Berlin,
  Heidelberg}.
\newblock
\showISBNx{3540741127}


\bibitem[\protect\citeauthoryear{Bryant, German, and Velev}{Bryant
  et~al\mbox{.}}{2001}]%
        {Bryant2001}
\bibfield{author}{\bibinfo{person}{Randal~E. Bryant}, \bibinfo{person}{Steven
  German}, {and} \bibinfo{person}{Miroslav~N. Velev}.}
  \bibinfo{year}{2001}\natexlab{}.
\newblock \showarticletitle{Processor Verification Using Efficient Reductions
  of the Logic of Uninterpreted Functions to Propositional Logic}.
\newblock \bibinfo{journal}{\emph{ACM Trans. Comput. Logic}}
  \bibinfo{volume}{2}, \bibinfo{number}{1} (\bibinfo{date}{Jan.}
  \bibinfo{year}{2001}), \bibinfo{pages}{93--134}.
\newblock
\showISSN{1529-3785}
\urldef\tempurl%
\url{https://doi.org/10.1145/371282.371364}
\showDOI{\tempurl}


\bibitem[\protect\citeauthoryear{Burch and Dill}{Burch and Dill}{1994}]%
        {Burch1994}
\bibfield{author}{\bibinfo{person}{Jerry~R. Burch} {and}
  \bibinfo{person}{David~L. Dill}.} \bibinfo{year}{1994}\natexlab{}.
\newblock \showarticletitle{Automatic Verification of Pipelined Microprocessor
  Control}. In \bibinfo{booktitle}{\emph{Proceedings of the 6th International
  Conference on Computer Aided Verification}} \emph{(\bibinfo{series}{CAV
  '94})}. \bibinfo{publisher}{Springer-Verlag}, \bibinfo{address}{London, UK,
  UK}, \bibinfo{pages}{68--80}.
\newblock
\showISBNx{3-540-58179-0}
\urldef\tempurl%
\url{http://dl.acm.org/citation.cfm?id=647763.735662}
\showURL{%
\tempurl}


\bibitem[\protect\citeauthoryear{Chatterjee, Goharshady, Ibsen-Jensen, and
  Pavlogiannis}{Chatterjee et~al\mbox{.}}{2016}]%
        {chatterjee2016}
\bibfield{author}{\bibinfo{person}{Krishnendu Chatterjee},
  \bibinfo{person}{Amir~Kafshdar Goharshady}, \bibinfo{person}{Rasmus
  Ibsen-Jensen}, {and} \bibinfo{person}{Andreas Pavlogiannis}.}
  \bibinfo{year}{2016}\natexlab{}.
\newblock \showarticletitle{Algorithms for Algebraic Path Properties in
  Concurrent Systems of Constant Treewidth Components}. In
  \bibinfo{booktitle}{\emph{Proceedings of the 43rd Annual ACM SIGPLAN-SIGACT
  Symposium on Principles of Programming Languages}}
  \emph{(\bibinfo{series}{POPL '16})}. \bibinfo{publisher}{ACM},
  \bibinfo{address}{New York, NY, USA}, \bibinfo{pages}{733--747}.
\newblock
\showISBNx{978-1-4503-3549-2}
\urldef\tempurl%
\url{https://doi.org/10.1145/2837614.2837624}
\showDOI{\tempurl}


\bibitem[\protect\citeauthoryear{Chatterjee, Ibsen-Jensen, Pavlogiannis, and
  Goyal}{Chatterjee et~al\mbox{.}}{2015}]%
        {chatterjee2015}
\bibfield{author}{\bibinfo{person}{Krishnendu Chatterjee},
  \bibinfo{person}{Rasmus Ibsen-Jensen}, \bibinfo{person}{Andreas
  Pavlogiannis}, {and} \bibinfo{person}{Prateesh Goyal}.}
  \bibinfo{year}{2015}\natexlab{}.
\newblock \showarticletitle{Faster Algorithms for Algebraic Path Properties in
  Recursive State Machines with Constant Treewidth}. In
  \bibinfo{booktitle}{\emph{Proceedings of the 42Nd Annual ACM SIGPLAN-SIGACT
  Symposium on Principles of Programming Languages}}
  \emph{(\bibinfo{series}{POPL '15})}. \bibinfo{publisher}{ACM},
  \bibinfo{address}{New York, NY, USA}, \bibinfo{pages}{97--109}.
\newblock
\showISBNx{978-1-4503-3300-9}
\urldef\tempurl%
\url{https://doi.org/10.1145/2676726.2676979}
\showDOI{\tempurl}


\bibitem[\protect\citeauthoryear{Courcelle and Engelfriet}{Courcelle and
  Engelfriet}{2012}]%
        {courcelle}
\bibfield{author}{\bibinfo{person}{Professor~Bruno Courcelle} {and}
  \bibinfo{person}{Dr~Joost Engelfriet}.} \bibinfo{year}{2012}\natexlab{}.
\newblock \bibinfo{booktitle}{\emph{Graph Structure and Monadic Second-Order
  Logic: A Language-Theoretic Approach} (\bibinfo{edition}{1st} ed.)}.
\newblock \bibinfo{publisher}{Cambridge University Press},
  \bibinfo{address}{New York, NY, USA}.
\newblock
\showISBNx{0521898331, 9780521898331}


\bibitem[\protect\citeauthoryear{Dillig, Dillig, and Aiken}{Dillig
  et~al\mbox{.}}{2011}]%
        {Dillig2011}
\bibfield{author}{\bibinfo{person}{Isil Dillig}, \bibinfo{person}{Thomas
  Dillig}, {and} \bibinfo{person}{Alex Aiken}.}
  \bibinfo{year}{2011}\natexlab{}.
\newblock \showarticletitle{Precise Reasoning for Programs Using Containers}.
  In \bibinfo{booktitle}{\emph{Proceedings of the 38th Annual ACM
  SIGPLAN-SIGACT Symposium on Principles of Programming Languages}}
  \emph{(\bibinfo{series}{POPL '11})}. \bibinfo{publisher}{ACM},
  \bibinfo{address}{New York, NY, USA}, \bibinfo{pages}{187--200}.
\newblock
\showISBNx{978-1-4503-0490-0}
\urldef\tempurl%
\url{https://doi.org/10.1145/1926385.1926407}
\showDOI{\tempurl}


\bibitem[\protect\citeauthoryear{Esparza, Hansel, Rossmanith, and
  Schwoon}{Esparza et~al\mbox{.}}{2000}]%
        {Esparza2000}
\bibfield{author}{\bibinfo{person}{Javier Esparza}, \bibinfo{person}{David
  Hansel}, \bibinfo{person}{Peter Rossmanith}, {and} \bibinfo{person}{Stefan
  Schwoon}.} \bibinfo{year}{2000}\natexlab{}.
\newblock \showarticletitle{Efficient Algorithms for Model Checking Pushdown
  Systems}. In \bibinfo{booktitle}{\emph{Proceedings of the 12th International
  Conference on Computer Aided Verification}} \emph{(\bibinfo{series}{CAV
  '00})}. \bibinfo{publisher}{Springer-Verlag}, \bibinfo{address}{London, UK,
  UK}, \bibinfo{pages}{232--247}.
\newblock
\showISBNx{3-540-67770-4}
\urldef\tempurl%
\url{http://dl.acm.org/citation.cfm?id=647769.734087}
\showURL{%
\tempurl}


\bibitem[\protect\citeauthoryear{Esparza and Knoop}{Esparza and Knoop}{1999}]%
        {EsparzaKnoopFOSSACS99}
\bibfield{author}{\bibinfo{person}{Javier Esparza} {and} \bibinfo{person}{Jens
  Knoop}.} \bibinfo{year}{1999}\natexlab{}.
\newblock \showarticletitle{An Automata-Theoretic Approach to Interprocedural
  Data-Flow Analysis}. In \bibinfo{booktitle}{\emph{Foundations of Software
  Science and Computation Structures}},
  \bibfield{editor}{\bibinfo{person}{Wolfgang Thomas}} (Ed.).
  \bibinfo{publisher}{Springer Berlin Heidelberg}, \bibinfo{address}{Berlin,
  Heidelberg}, \bibinfo{pages}{14--30}.
\newblock
\showISBNx{978-3-540-49019-7}


\bibitem[\protect\citeauthoryear{Farzan, Kincaid, and Podelski}{Farzan
  et~al\mbox{.}}{2014}]%
        {Farzan2014}
\bibfield{author}{\bibinfo{person}{Azadeh Farzan}, \bibinfo{person}{Zachary
  Kincaid}, {and} \bibinfo{person}{Andreas Podelski}.}
  \bibinfo{year}{2014}\natexlab{}.
\newblock \showarticletitle{Proofs That Count}. In
  \bibinfo{booktitle}{\emph{Proceedings of the 41st ACM SIGPLAN-SIGACT
  Symposium on Principles of Programming Languages}}
  \emph{(\bibinfo{series}{POPL '14})}. \bibinfo{publisher}{ACM},
  \bibinfo{address}{New York, NY, USA}, \bibinfo{pages}{151--164}.
\newblock
\showISBNx{978-1-4503-2544-8}
\urldef\tempurl%
\url{https://doi.org/10.1145/2535838.2535885}
\showDOI{\tempurl}


\bibitem[\protect\citeauthoryear{Farzan, Kincaid, and Podelski}{Farzan
  et~al\mbox{.}}{2015}]%
        {Farzan2015}
\bibfield{author}{\bibinfo{person}{Azadeh Farzan}, \bibinfo{person}{Zachary
  Kincaid}, {and} \bibinfo{person}{Andreas Podelski}.}
  \bibinfo{year}{2015}\natexlab{}.
\newblock \showarticletitle{Proof Spaces for Unbounded Parallelism}. In
  \bibinfo{booktitle}{\emph{Proceedings of the 42Nd Annual ACM SIGPLAN-SIGACT
  Symposium on Principles of Programming Languages}}
  \emph{(\bibinfo{series}{POPL '15})}. \bibinfo{publisher}{ACM},
  \bibinfo{address}{New York, NY, USA}, \bibinfo{pages}{407--420}.
\newblock
\showISBNx{978-1-4503-3300-9}
\urldef\tempurl%
\url{https://doi.org/10.1145/2676726.2677012}
\showDOI{\tempurl}


\bibitem[\protect\citeauthoryear{Godefroid and Yannakakis}{Godefroid and
  Yannakakis}{2013}]%
        {Godefroid2013}
\bibfield{author}{\bibinfo{person}{Patrice Godefroid} {and}
  \bibinfo{person}{Mihalis Yannakakis}.} \bibinfo{year}{2013}\natexlab{}.
\newblock \showarticletitle{Analysis of Boolean Programs}. In
  \bibinfo{booktitle}{\emph{Proceedings of the 19th International Conference on
  Tools and Algorithms for the Construction and Analysis of Systems}}
  \emph{(\bibinfo{series}{TACAS'13})}. \bibinfo{publisher}{Springer-Verlag},
  \bibinfo{address}{Berlin, Heidelberg}, \bibinfo{pages}{214--229}.
\newblock
\showISBNx{978-3-642-36741-0}
\urldef\tempurl%
\url{https://doi.org/10.1007/978-3-642-36742-7_16}
\showDOI{\tempurl}


\bibitem[\protect\citeauthoryear{Godoy and Tiwari}{Godoy and Tiwari}{2009}]%
        {godoy09}
\bibfield{author}{\bibinfo{person}{Guillem Godoy} {and} \bibinfo{person}{Ashish
  Tiwari}.} \bibinfo{year}{2009}\natexlab{}.
\newblock \showarticletitle{Invariant Checking for Programs with Procedure
  Calls}. In \bibinfo{booktitle}{\emph{Proceedings of the 16th International
  Symposium on Static Analysis}} \emph{(\bibinfo{series}{SAS '09})}.
  \bibinfo{publisher}{Springer-Verlag}, \bibinfo{address}{Berlin, Heidelberg},
  \bibinfo{pages}{326--342}.
\newblock
\showISBNx{978-3-642-03236-3}
\urldef\tempurl%
\url{https://doi.org/10.1007/978-3-642-03237-0_22}
\showDOI{\tempurl}


\bibitem[\protect\citeauthoryear{Granger}{Granger}{1991}]%
        {Granger1991}
\bibfield{author}{\bibinfo{person}{Philippe Granger}.}
  \bibinfo{year}{1991}\natexlab{}.
\newblock \showarticletitle{Static Analysis of Linear Congruence Equalities
  Among Variables of a Program}. In \bibinfo{booktitle}{\emph{Proceedings of
  the International Joint Conference on Theory and Practice of Software
  Development on Colloquium on Trees in Algebra and Programming (CAAP '91): Vol
  1}} \emph{(\bibinfo{series}{TAPSOFT '91})}.
  \bibinfo{publisher}{Springer-Verlag New York, Inc.}, \bibinfo{address}{New
  York, NY, USA}, \bibinfo{pages}{169--192}.
\newblock
\showISBNx{3-540-53982-4}
\urldef\tempurl%
\url{http://dl.acm.org/citation.cfm?id=111310.111320}
\showURL{%
\tempurl}


\bibitem[\protect\citeauthoryear{Gulwani and Necula}{Gulwani and
  Necula}{2004a}]%
        {Gulwani2004}
\bibfield{author}{\bibinfo{person}{Sumit Gulwani} {and}
  \bibinfo{person}{George~C. Necula}.} \bibinfo{year}{2004}\natexlab{a}.
\newblock \showarticletitle{Global Value Numbering Using Random
  Interpretation}. In \bibinfo{booktitle}{\emph{Proceedings of the 31st ACM
  SIGPLAN-SIGACT Symposium on Principles of Programming Languages}}
  \emph{(\bibinfo{series}{POPL '04})}. \bibinfo{publisher}{ACM},
  \bibinfo{address}{New York, NY, USA}, \bibinfo{pages}{342--352}.
\newblock
\showISBNx{1-58113-729-X}
\urldef\tempurl%
\url{https://doi.org/10.1145/964001.964030}
\showDOI{\tempurl}


\bibitem[\protect\citeauthoryear{Gulwani and Necula}{Gulwani and
  Necula}{2004b}]%
        {gulwani2004polynomial}
\bibfield{author}{\bibinfo{person}{Sumit Gulwani} {and}
  \bibinfo{person}{George~C Necula}.} \bibinfo{year}{2004}\natexlab{b}.
\newblock \showarticletitle{A polynomial-time algorithm for global value
  numbering}. In \bibinfo{booktitle}{\emph{International Static Analysis
  Symposium}}. Springer, \bibinfo{pages}{212--227}.
\newblock


\bibitem[\protect\citeauthoryear{Gulwani and Tiwari}{Gulwani and
  Tiwari}{2006}]%
        {Gulwani2006}
\bibfield{author}{\bibinfo{person}{Sumit Gulwani} {and} \bibinfo{person}{Ashish
  Tiwari}.} \bibinfo{year}{2006}\natexlab{}.
\newblock \showarticletitle{Assertion Checking over Combined Abstraction of
  Linear Arithmetic and Uninterpreted Functions}. In
  \bibinfo{booktitle}{\emph{Proceedings of the 15th European Conference on
  Programming Languages and Systems}} \emph{(\bibinfo{series}{ESOP'06})}.
  \bibinfo{publisher}{Springer-Verlag}, \bibinfo{address}{Berlin, Heidelberg},
  \bibinfo{pages}{279--293}.
\newblock
\showISBNx{3-540-33095-X, 978-3-540-33095-0}
\urldef\tempurl%
\url{https://doi.org/10.1007/11693024_19}
\showDOI{\tempurl}


\bibitem[\protect\citeauthoryear{Heizmann, Hoenicke, and Podelski}{Heizmann
  et~al\mbox{.}}{2009}]%
        {Heizmann2009SAS}
\bibfield{author}{\bibinfo{person}{Matthias Heizmann}, \bibinfo{person}{Jochen
  Hoenicke}, {and} \bibinfo{person}{Andreas Podelski}.}
  \bibinfo{year}{2009}\natexlab{}.
\newblock \showarticletitle{Refinement of Trace Abstraction}. In
  \bibinfo{booktitle}{\emph{Proceedings of the 16th International Symposium on
  Static Analysis}} \emph{(\bibinfo{series}{SAS '09})}.
  \bibinfo{publisher}{Springer-Verlag}, \bibinfo{address}{Berlin, Heidelberg},
  \bibinfo{pages}{69--85}.
\newblock
\showISBNx{978-3-642-03236-3}
\urldef\tempurl%
\url{https://doi.org/10.1007/978-3-642-03237-0_7}
\showDOI{\tempurl}


\bibitem[\protect\citeauthoryear{Heizmann, Hoenicke, and Podelski}{Heizmann
  et~al\mbox{.}}{2010}]%
        {Heizmann2010}
\bibfield{author}{\bibinfo{person}{Matthias Heizmann}, \bibinfo{person}{Jochen
  Hoenicke}, {and} \bibinfo{person}{Andreas Podelski}.}
  \bibinfo{year}{2010}\natexlab{}.
\newblock \showarticletitle{Nested Interpolants}. In
  \bibinfo{booktitle}{\emph{Proceedings of the 37th Annual ACM SIGPLAN-SIGACT
  Symposium on Principles of Programming Languages}}
  \emph{(\bibinfo{series}{POPL '10})}. \bibinfo{publisher}{ACM},
  \bibinfo{address}{New York, NY, USA}, \bibinfo{pages}{471--482}.
\newblock
\showISBNx{978-1-60558-479-9}
\urldef\tempurl%
\url{https://doi.org/10.1145/1706299.1706353}
\showDOI{\tempurl}


\bibitem[\protect\citeauthoryear{Heizmann, Hoenicke, and Podelski}{Heizmann
  et~al\mbox{.}}{2013}]%
        {Matthias2013}
\bibfield{author}{\bibinfo{person}{Matthias Heizmann}, \bibinfo{person}{Jochen
  Hoenicke}, {and} \bibinfo{person}{Andreas Podelski}.}
  \bibinfo{year}{2013}\natexlab{}.
\newblock \showarticletitle{Software Model Checking for People Who Love
  Automata}. In \bibinfo{booktitle}{\emph{Computer Aided Verification}},
  \bibfield{editor}{\bibinfo{person}{Natasha Sharygina} {and}
  \bibinfo{person}{Helmut Veith}} (Eds.). \bibinfo{publisher}{Springer Berlin
  Heidelberg}, \bibinfo{address}{Berlin, Heidelberg}, \bibinfo{pages}{36--52}.
\newblock
\showISBNx{978-3-642-39799-8}


\bibitem[\protect\citeauthoryear{Karp and Miller}{Karp and Miller}{1969}]%
        {Karp1969}
\bibfield{author}{\bibinfo{person}{Richard~M. Karp} {and}
  \bibinfo{person}{Raymond~E. Miller}.} \bibinfo{year}{1969}\natexlab{}.
\newblock \showarticletitle{Parallel Program Schemata}.
\newblock \bibinfo{journal}{\emph{J. Comput. Syst. Sci.}} \bibinfo{volume}{3},
  \bibinfo{number}{2} (\bibinfo{date}{May} \bibinfo{year}{1969}),
  \bibinfo{pages}{147--195}.
\newblock
\showISSN{0022-0000}
\urldef\tempurl%
\url{https://doi.org/10.1016/S0022-0000(69)80011-5}
\showDOI{\tempurl}


\bibitem[\protect\citeauthoryear{Karr}{Karr}{1976}]%
        {Karr1976}
\bibfield{author}{\bibinfo{person}{Michael Karr}.}
  \bibinfo{year}{1976}\natexlab{}.
\newblock \showarticletitle{Affine Relationships Among Variables of a Program}.
\newblock \bibinfo{journal}{\emph{Acta Inf.}} \bibinfo{volume}{6},
  \bibinfo{number}{2} (\bibinfo{date}{June} \bibinfo{year}{1976}),
  \bibinfo{pages}{133--151}.
\newblock
\showISSN{0001-5903}
\urldef\tempurl%
\url{https://doi.org/10.1007/BF00268497}
\showDOI{\tempurl}


\bibitem[\protect\citeauthoryear{Kosaraju}{Kosaraju}{1982}]%
        {Kosaraju1982}
\bibfield{author}{\bibinfo{person}{S.~Rao Kosaraju}.}
  \bibinfo{year}{1982}\natexlab{}.
\newblock \showarticletitle{Decidability of Reachability in Vector Addition
  Systems (Preliminary Version)}. In \bibinfo{booktitle}{\emph{Proceedings of
  the Fourteenth Annual ACM Symposium on Theory of Computing}}
  \emph{(\bibinfo{series}{STOC '82})}. \bibinfo{publisher}{ACM},
  \bibinfo{address}{New York, NY, USA}, \bibinfo{pages}{267--281}.
\newblock
\showISBNx{0-89791-070-2}
\urldef\tempurl%
\url{https://doi.org/10.1145/800070.802201}
\showDOI{\tempurl}


\bibitem[\protect\citeauthoryear{L\"{o}ding, Madhusudan, and
  Pe\~{n}a}{L\"{o}ding et~al\mbox{.}}{2017}]%
        {natproofs2017}
\bibfield{author}{\bibinfo{person}{Christof L\"{o}ding}, \bibinfo{person}{P.
  Madhusudan}, {and} \bibinfo{person}{Lucas Pe\~{n}a}.}
  \bibinfo{year}{2017}\natexlab{}.
\newblock \showarticletitle{Foundations for Natural Proofs and Quantifier
  Instantiation}.
\newblock \bibinfo{journal}{\emph{Proc. ACM Program. Lang.}}
  \bibinfo{volume}{2}, \bibinfo{number}{POPL}, Article \bibinfo{articleno}{10}
  (\bibinfo{date}{Dec.} \bibinfo{year}{2017}), \bibinfo{numpages}{30}~pages.
\newblock
\showISSN{2475-1421}
\urldef\tempurl%
\url{https://doi.org/10.1145/3158098}
\showDOI{\tempurl}


\bibitem[\protect\citeauthoryear{Lopes and Monteiro}{Lopes and
  Monteiro}{2016}]%
        {Lopes2016}
\bibfield{author}{\bibinfo{person}{Nuno~P. Lopes} {and}
  \bibinfo{person}{Jos{\'e} Monteiro}.} \bibinfo{year}{2016}\natexlab{}.
\newblock \showarticletitle{Automatic equivalence checking of programs with
  uninterpreted functions and integer arithmetic}.
\newblock \bibinfo{journal}{\emph{International Journal on Software Tools for
  Technology Transfer}} \bibinfo{volume}{18}, \bibinfo{number}{4}
  (\bibinfo{date}{01 Aug} \bibinfo{year}{2016}), \bibinfo{pages}{359--374}.
\newblock
\showISSN{1433-2787}
\urldef\tempurl%
\url{https://doi.org/10.1007/s10009-015-0366-1}
\showDOI{\tempurl}


\bibitem[\protect\citeauthoryear{Madhusudan and Parlato}{Madhusudan and
  Parlato}{2011}]%
        {madhu2011}
\bibfield{author}{\bibinfo{person}{P. Madhusudan} {and}
  \bibinfo{person}{Gennaro Parlato}.} \bibinfo{year}{2011}\natexlab{}.
\newblock \showarticletitle{The Tree Width of Auxiliary Storage}. In
  \bibinfo{booktitle}{\emph{Proceedings of the 38th Annual ACM SIGPLAN-SIGACT
  Symposium on Principles of Programming Languages}}
  \emph{(\bibinfo{series}{POPL '11})}. \bibinfo{publisher}{ACM},
  \bibinfo{address}{New York, NY, USA}, \bibinfo{pages}{283--294}.
\newblock
\showISBNx{978-1-4503-0490-0}
\urldef\tempurl%
\url{https://doi.org/10.1145/1926385.1926419}
\showDOI{\tempurl}


\bibitem[\protect\citeauthoryear{Mayr}{Mayr}{1981}]%
        {Mayr1981}
\bibfield{author}{\bibinfo{person}{Ernst~W. Mayr}.}
  \bibinfo{year}{1981}\natexlab{}.
\newblock \showarticletitle{An Algorithm for the General Petri Net Reachability
  Problem}. In \bibinfo{booktitle}{\emph{Proceedings of the Thirteenth Annual
  ACM Symposium on Theory of Computing}} \emph{(\bibinfo{series}{STOC '81})}.
  \bibinfo{publisher}{ACM}, \bibinfo{address}{New York, NY, USA},
  \bibinfo{pages}{238--246}.
\newblock
\urldef\tempurl%
\url{https://doi.org/10.1145/800076.802477}
\showDOI{\tempurl}


\bibitem[\protect\citeauthoryear{M{\"u}ller-Olm and Seidl}{M{\"u}ller-Olm and
  Seidl}{2004}]%
        {muller-olm2004}
\bibfield{author}{\bibinfo{person}{Markus M{\"u}ller-Olm} {and}
  \bibinfo{person}{Helmut Seidl}.} \bibinfo{year}{2004}\natexlab{}.
\newblock \showarticletitle{A Note on Karr's Algorithm}. In
  \bibinfo{booktitle}{\emph{Automata, Languages and Programming}},
  \bibfield{editor}{\bibinfo{person}{Josep D{\'i}az}, \bibinfo{person}{Juhani
  Karhum{\"a}ki}, \bibinfo{person}{Arto Lepist{\"o}}, {and}
  \bibinfo{person}{Donald Sannella}} (Eds.). \bibinfo{publisher}{Springer
  Berlin Heidelberg}, \bibinfo{address}{Berlin, Heidelberg},
  \bibinfo{pages}{1016--1028}.
\newblock


\bibitem[\protect\citeauthoryear{M\"{u}ller-Olm and Seidl}{M\"{u}ller-Olm and
  Seidl}{2005}]%
        {Muller-Olm2005}
\bibfield{author}{\bibinfo{person}{Markus M\"{u}ller-Olm} {and}
  \bibinfo{person}{Helmut Seidl}.} \bibinfo{year}{2005}\natexlab{}.
\newblock \showarticletitle{A Generic Framework for Interprocedural Analysis of
  Numerical Properties}. In \bibinfo{booktitle}{\emph{Proceedings of the 12th
  International Conference on Static Analysis}}
  \emph{(\bibinfo{series}{SAS'05})}. \bibinfo{publisher}{Springer-Verlag},
  \bibinfo{address}{Berlin, Heidelberg}, \bibinfo{pages}{235--250}.
\newblock
\showISBNx{3-540-28584-9, 978-3-540-28584-7}
\urldef\tempurl%
\url{https://doi.org/10.1007/11547662_17}
\showDOI{\tempurl}


\bibitem[\protect\citeauthoryear{Pek, Qiu, and Madhusudan}{Pek
  et~al\mbox{.}}{2014}]%
        {natproofs2014}
\bibfield{author}{\bibinfo{person}{Edgar Pek}, \bibinfo{person}{Xiaokang Qiu},
  {and} \bibinfo{person}{P. Madhusudan}.} \bibinfo{year}{2014}\natexlab{}.
\newblock \showarticletitle{Natural Proofs for Data Structure Manipulation in C
  Using Separation Logic}. In \bibinfo{booktitle}{\emph{Proceedings of the 35th
  ACM SIGPLAN Conference on Programming Language Design and Implementation}}
  \emph{(\bibinfo{series}{PLDI '14})}. \bibinfo{publisher}{ACM},
  \bibinfo{address}{New York, NY, USA}, \bibinfo{pages}{440--451}.
\newblock
\showISBNx{978-1-4503-2784-8}
\urldef\tempurl%
\url{https://doi.org/10.1145/2594291.2594325}
\showDOI{\tempurl}


\bibitem[\protect\citeauthoryear{Qiu, Garg, \c{S}tef\u{a}nescu, and
  Madhusudan}{Qiu et~al\mbox{.}}{2013}]%
        {natproofs2013}
\bibfield{author}{\bibinfo{person}{Xiaokang Qiu}, \bibinfo{person}{Pranav
  Garg}, \bibinfo{person}{Andrei \c{S}tef\u{a}nescu}, {and}
  \bibinfo{person}{Parthasarathy Madhusudan}.} \bibinfo{year}{2013}\natexlab{}.
\newblock \showarticletitle{Natural Proofs for Structure, Data, and
  Separation}. In \bibinfo{booktitle}{\emph{Proceedings of the 34th ACM SIGPLAN
  Conference on Programming Language Design and Implementation}}
  \emph{(\bibinfo{series}{PLDI '13})}. \bibinfo{publisher}{ACM},
  \bibinfo{address}{New York, NY, USA}, \bibinfo{pages}{231--242}.
\newblock
\showISBNx{978-1-4503-2014-6}
\urldef\tempurl%
\url{https://doi.org/10.1145/2491956.2462169}
\showDOI{\tempurl}


\bibitem[\protect\citeauthoryear{Robertson and Seymour}{Robertson and
  Seymour}{1983}]%
        {robertson1983graph}
\bibfield{author}{\bibinfo{person}{Neil Robertson} {and}
  \bibinfo{person}{Paul~D Seymour}.} \bibinfo{year}{1983}\natexlab{}.
\newblock \showarticletitle{Graph minors. I. Excluding a forest}.
\newblock \bibinfo{journal}{\emph{Journal of Combinatorial Theory, Series B}}
  \bibinfo{volume}{35}, \bibinfo{number}{1} (\bibinfo{year}{1983}),
  \bibinfo{pages}{39--61}.
\newblock


\bibitem[\protect\citeauthoryear{Schwoon}{Schwoon}{2002}]%
        {schwoon-phd02}
\bibfield{author}{\bibinfo{person}{Stefan Schwoon}.}
  \bibinfo{year}{2002}\natexlab{}.
\newblock \emph{\bibinfo{title}{Model-Checking Pushdown Systems}}.
\newblock {Ph.}{D.} {T}hesis. \bibinfo{school}{Technische Universit{\"a}t
  M{\"u}nchen}.
\newblock
\urldef\tempurl%
\url{http://www.lsv.ens-cachan.fr/Publis/PAPERS/PDF/schwoon-phd02.pdf}
\showURL{%
\tempurl}


\bibitem[\protect\citeauthoryear{Seese}{Seese}{1991}]%
        {seese}
\bibfield{author}{\bibinfo{person}{D. Seese}.} \bibinfo{year}{1991}\natexlab{}.
\newblock \showarticletitle{The structure of the models of decidable monadic
  theories of graphs}.
\newblock \bibinfo{journal}{\emph{Annals of Pure and Applied Logic}}
  \bibinfo{volume}{53}, \bibinfo{number}{2} (\bibinfo{year}{1991}),
  \bibinfo{pages}{169 -- 195}.
\newblock
\showISSN{0168-0072}
\urldef\tempurl%
\url{https://doi.org/10.1016/0168-0072(91)90054-P}
\showDOI{\tempurl}


\end{thebibliography}

\end{document}